\documentclass[draftclsnofoot]{IEEEtran}
\onecolumn
\usepackage{latexsym}
\usepackage{cite}
\usepackage{color}
\usepackage{bm}
\usepackage{amsmath, amssymb}
\usepackage{lipsum}
\usepackage[dvips]{graphics}
\usepackage{graphicx}
\usepackage{psfrag}
 \allowdisplaybreaks

\newtheorem{definition}{Definition}

\newtheorem{theorem}{Theorem}
\newtheorem{lemma}[theorem]{Lemma}
\newtheorem{remark}{Remark}

\newtheorem{corollary}[theorem]{Corollary}

\newcommand\blfootnote[1]{%
  \begingroup
  \renewcommand\thefootnote{}\footnote{#1}%
  \addtocounter{footnote}{-1}%
  \endgroup
}

\begin{document}
\title{On State Dependent Broadcast Channels with Cooperation}
\author{Lior Dikstein, Haim H. Permuter and Yossef Steinberg}
\maketitle

\maketitle
\blfootnote{L. Dikstein and H. Permuter are with the department of Electrical and Computer Engineering, Ben-Gurion University of the Negev, Beer-Sheva, Israel (liordik@post.bgu.ac.il, haimp@bgu.ac.il). Y. Steinberg is with the Department of Electrical Engineering, Technion, Haifa, Israel (ysteinbe@ee.technion.ac.il). This paper was presented in part at the 2013 51st Annual Allerton Conference on Communication, Control, and Computing (Allerton). This work was supported by the Israel Science Foundation (grant no. 684/11).}

\begin{abstract}
In this paper, we investigate problems of communication over physically degraded, state-dependent broadcast channels (BCs) with cooperating decoders. Two different setups are considered and their capacity regions are characterized. First, we study a setting in which one decoder can use a finite capacity link to send the other decoder information regarding the messages or the channel states. In this scenario we analyze two cases: one where noncausal state information is available to the encoder and the strong decoder and the other where state information is available only to the encoder in a causal manner. Second, we examine a setting in which the cooperation between the decoders is limited to taking place before the outputs of the channel are given. In this case, one decoder, which is informed of the state sequence noncausally, can cooperate only to send the other decoder rate-limited information about the state sequence. The proofs of the capacity regions introduce a new method of coding for channels with cooperation between different users, where we exploit the link between the decoders for multiple-binning. Finally, we discuss the optimality of using rate splitting techniques when coding for cooperative BCs. In particular, we show that rate splitting is not necessarily optimal when coding for cooperative BCs by solving an example in which our method of coding outperforms rate splitting.
\end{abstract}


\section{Introduction}
\par Classical broadcast channels (BCs) adequately model a variety of practical communication scenarios, such as cellular systems, Wi-Fi routers and digital TV broadcasting. However, with the rapid growth of wireless networking, it is necessary to expand the study of such channels and to consider more complex settings that can more accurately describe a wider range of scenarios. Some of these important extensions include settings in which the BC is state dependent. Wireless channels with fading, jamming or interference in the transmission are but a few examples that can be modeled by state-dependent settings. Other important extensions include cooperation between different nodes in a network, where the nodes help one another in decoding. These settings can, inter alia, describe sensor networks, in which a large number of nodes conduct measurements on some ongoing process. When such measurements are correlated, cooperation between nodes can assist them in decoding. Some practical sensor network applications include surveillance, health monitoring and environmental detection. Therefore, the results presented in this work will contribute to meeting the growing need to find the fundamental limits of such important communication scenarios.

\par The most general form of the BC, in which a single source attempts to communicate simultaneously to two or more receivers, was introduced by Cover in \cite{CoverBroadcast}. Following his work, the capacity of the degraded BC was characterized by Bergmans \cite{BergmansBroadcast} and Gallager \cite{GallagerBroadcast}. In the degraded BC setting, one receiver is statistically stronger than the other. This scenario can, for instance, model TV broadcasts, where some users consume high definition media, while other users watch the same broadcast with lower resolution. 
 In \cite{ElGammalClassOfBroadcast}, El-Gamal expanded the capacity result for the degraded BC and in \cite{ElGammalBroadcastFeedback}-\cite{ElGammalPhysically} discussed the two-user BC with and without feedback. Later, in \cite{ElGamal: FeedbackCapacity} El-Gamal showed that the capacity of the two-user physically degraded BC does not change with feedback.

\par State-dependent channels were first introduced by Shannon \cite{Shannon58}, who characterized the capacity for the case where a single user channel, $P_{Y|X,S}$, is controlled by a random parameter $S$, and the state information sequence up to time $i$, $s^i$, is known causally at the encoder. The case in which the full realization of the channel states, $s^n$, is known noncausally at the encoder was presented by Kuznetsov and Tsybakov \cite{{KuzTsy74}}, in the context of coding for a storage unit, and similar cases were studied by Heegard and El Gamal in \cite{HeegardElGamal_state_encoded83}. The capacity of the channel for this case was fully solved by  Gel'fand and Pinsker in \cite{GePi80}.
In recent years, growing interest in state-dependent channels has resulted in many studies on multi-user settings. Some examples considering multiple access channels (MAC) include the works of Lapidoth and Steinberg \cite{Lapidoth2010MAC}, \cite{Lapidoth2013MAC}, Piantanida, Zaidi and Shamai \cite{ZaidiPiantanidaShamai}, Somekh-Baruch, Shamai and Verdu \cite{BaruchShamaiVerduMAC} and many more. In the case of BCs, Steinberg studied the degraded, state-dependent BC in \cite{steinberg_causal_BC_2005}. Inner and outer bounds were derived for the case in which the state information is known noncausally at the encoder and the capacity region was found for the case in which the states are known causally at the encoder or known noncausally at both the encoder and the strong decoder. Our channel setting with cooperation naturally extends this model, and the capacity results of this paper generalize the capacity regions found in \cite{steinberg_causal_BC_2005}.

\par Other important settings for state-dependent channels are cases where only rate-limited state information is available at the encoder or decoder. In many practical systems, information on the channel is not available freely. Thus, to provide side information on the channel's states to the different users, we must allocate resources such as time slots, bandwidth and memory. Heegard and El Gamal \cite{HeegardElGamal_state_encoded83} presented a model of a state-dependent channel, where the transmitter is informed of the state information at a rate-limited to $R_e$ and the receiver is informed of the state information at a rate-limited to $R_d$. Cover and Chiang \cite{CoCh02} extended the Gel'fand-Pinsker problem to the case where both the encoder and the decoder are provided with different, but correlated, partial side information. Steinberg in \cite{Steinberg08RateLimitedAtDecoder} derived the capacity of a channel where the state information is known fully to the encoder in a noncausal manner and is available rate-limited at the decoder. Coding for such a channel involves solving a Gel'fand-Pinsker problem as well as a Wyner-Ziv problem \cite{Wyner_ziv76} simultaneously. An extension of this setting to a degraded, state-dependent BC is introduced in this work.

\par Cooperation between different users in multi-user channels was first considered for MACs in the works of Willems \cite{Willems83_cooperating}. Further studies involving cooperation between encoders in MACs include papers such as the work of Willems and van der Meulen \cite{Willems85_cribbing_encoders}, and later \cite{Simeone:2009} and \cite{AmirSteinberg13}. A setting in which cooperation between users takes place in a state-dependent MAC, where partial state information is known to each user and full state information is known at the decoder, was treated by  Shamai, Permuter and Somekh-Baruch in \cite{PermuterShamaiSomekh_MessageStateCooperation}.

\par The notion of cooperation between receivers in a BC was proposed by Dabora and Servetto \cite{R.Dabora_S.ServettoRS}, where the capacity region for the cooperative physically degraded BC is characterized. Simultaneously, Liang and Veeravalli independently examined a more general problem of a relay-BC in \cite{LiangVeeravalli07RelayBroadcast}. The direct part of the proof for the capacity region of the cooperative physically degraded BC, combines methods of broadcast coding together with a code construction for the degraded relay channel. The BC setting we present in this work generalizes the model of Dabora and Servetto, and therefore our capacity result generalizes the result of \cite{R.Dabora_S.ServettoRS}. Moreover, the coding scheme we propose, that achieves the capacity of the more general physically degraded state-dependent cooperative BC, is fundamentally different and in some sense simpler, since it uses binning instead of block Markov coding.

\par In this work, we consider several scenarios of cooperation between decoders for physically degraded, state-dependent (PDSD) BCs. First, we consider a setting in which there is no constraint on when the cooperation link between the decoders, $C_{12}$, should be used. For this setting, we characterize the capacity region for the noncausal case, where noncausal state information is known at the encoder and the strong decoder, and for the causal case, where causal state information is available only at the encoder. The proof proposes a new coding method for channels with cooperating users, using multiple-binning. We suggest dividing the weak decoder's message set into bins, where the number of bins is determined by the capacity link between the decoders. The strong decoder will use the cooperation link to send the weak decoder the bin number containing its message, and hence narrow down the search from the entire message set, to the message set of that bin alone. This scheme increases the rate of the weak user.
\par The optimal schemes of the first scenario use the cooperation link between the decoders, $C_{12}$, solely for sending additional information about the messages, i.e., information about the state sequence is not sent explicitly via $C_{12}$. The second setting we consider is a case in which the cooperation link $C_{12}$ can be used only before the outputs are given. In such a case, the strong decoder can use the cooperation link only to convey rate-limited state information to the weaker user. This setting can be regarded as a broadcast extension of the results in \cite{Steinberg08RateLimitedAtDecoder}. The capacity region for this case is derived by using the methods we developed when solving the first scenario, combined with Wyner-Ziv coding.


\par Another interesting result involves the use of rate splitting techniques when coding for cooperative BC. In MACs, rate splitting, which is the most common coding method for dealing with cooperation, has been used to achieve the capacity for most settings that have been solved. Thus, a first guess would be to use this method when coding for cooperative BCs. However, we show that rate splitting schemes are not necessarily optimal for BCs. Moreover, we demonstrate that out method of coding, using binning, strictly outperforms rate splitting. An example for which rate splitting is shown to be suboptimal is the binary symmetric BC with cooperation.

\par The remainder of the paper is organized as follows. Section \ref{SectionNotation} presents the mathematical notation and definitions used in this paper. In Section \ref{SectionMainResults}, all the main results are stated, which include the capacity regions of three PDSD BC settings. Section \ref{SectionCapacity1} is devoted to the noncausal PDSD BC and a discussion of special cases, Section \ref{SectionCapacity2} is dedicated to the causal PDSD BC, and Section \ref{SectionCapacity3} is dedicated to the PDSD BC with rate-limited state information. In Section \ref{SectionRS}, we discuss the optimality of using rate splitting methods when dealing with cooperating BC, and in Section \ref{SubSectionBinarySymetric} we give an example of a cooperative BC in which rate splitting is suboptimal. Finally, proofs are given in Section \ref{SectionProofs}.


\section{Notation and Problem Definition}\label{SectionNotation}

\begin{figure}[h!]
    \begin{center}
        \begin{psfrags}
            \psfragscanon
            \psfrag{B}[][][0.9]{$A$}
            \psfrag{D}[][][0.9]{$B$}
            \psfrag{C}[][][0.9]{$S^{i}$}
            \psfrag{X}[][][0.9]{ $S^n$}
            \psfrag{I}[][][0.9]{\ \ \ \ $M_Z$}
            \psfrag{Z}[][][0.9]{\ \ \ \ $M_Y$}
            \psfrag{J}[][][0.9]{\ \ \ \ \ \ \ \ \ \  Encoder}
            \psfrag{K}[][][0.9]{\ \  $X^n$}
            \psfrag{S}[][][0.9]{\ \ \ \ \ \ \ \ \ \ \ \ \ \ \ \ \ \ \  $P_{Y,Z|X,S}=$}
            \psfrag{L}[][][0.9]{\ \ \ \ \ \ \ \ \ \ \ \ \ \  \ \ \ \ \ $P_{Y|X,S}P_{Z|Y}$}
            \psfrag{M}[][][0.9]{\ \ \ \ \ \ \  $Y^n$}
            \psfrag{N}[][][0.9]{\ \ \ \ \ \ \  $Z^n$}
            \psfrag{O}[][][0.9]{\ \ \ \ \ \ \ \ \ \ \ \  Decoder Y}
            \psfrag{P}[][][0.9]{\ \ \ \ \ \ \ \ \ \ \ \  Decoder Z}
            \psfrag{A}[][][0.9]{\ \ $C_{12}$}
            \psfrag{Q}[][][0.9]{$\hat{M}_Y$}
            \psfrag{R}[][][0.9]{$\hat{M}_Z$}
            \centerline{\includegraphics[scale = .6]{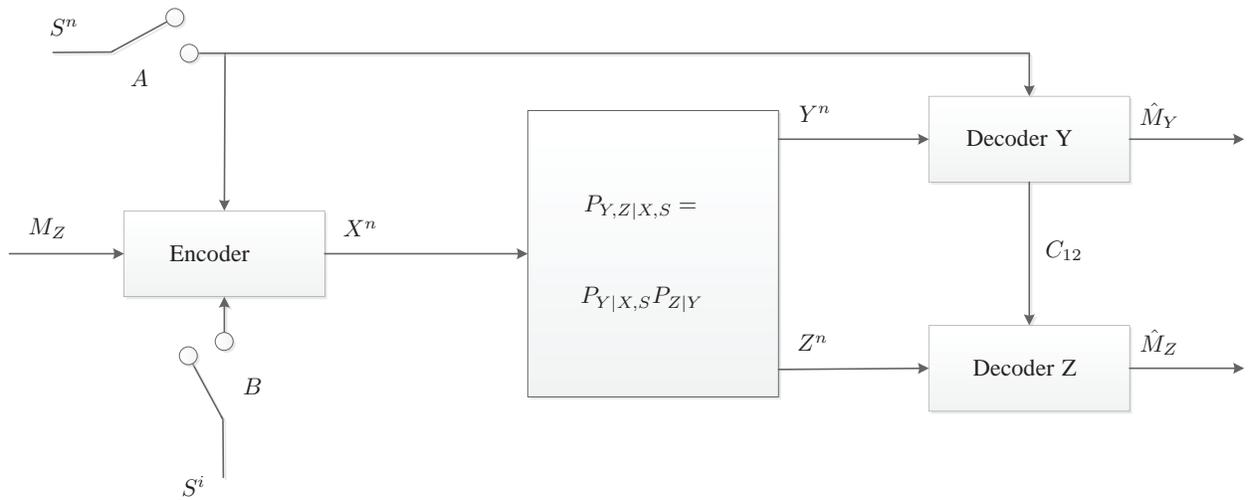}}
            \caption{The physically degraded, state-dependent BC with cooperating decoders. When considering cooperation between decoders in a physically degrade BC setting, only a cooperation link from the strong decoder to the weak decoder will contribute to increasing the rate.} \label{channel1}
            \psfragscanoff
        \end{psfrags}
     \end{center}
 \end{figure}
\par In this paper, random variables are denoted with upper case letters, deterministic realizations or specific values are denoted with lower case letters, and calligraphic letters denote the alphabets of the random variables. Vectors of $n$ elements, $(x_1,x_2,...,x_n)$ are denoted as $x^n$, and $x_i^j$ denotes the $i-j+1$-tuple $(x_i,x_{i+1},...,x_j)$ when $j\geq i$ and an empty set otherwise. The probability distribution function of $X$, the joint distribution function of $X$ and $Y$, and the conditional distribution of $X$ given $Y$ are denoted by $P_{X}$, $P_{X,Y}$ and $P_{X|Y}$, respectively.

\par A PDSD BC, $(\mathcal{S},P_S(s),\mathcal{X},\mathcal{Y},\mathcal{Z},P_{Y,Z|X,S}(y,z|x,s))$, illustrated in Fig. \ref{channel1}, is a channel with input alphabet $\mathcal{X}$, output alphabet $\mathcal{Y}\times\mathcal{Z}$ and a state space $\mathcal{S}$. The encoder selects a channel input sequence, $X^n = X^n(M_Z,M_Y,S^n)$. The outputs of the channel at Decoder Y and Decoder Z are denoted $Y^n$ and $Z^n$, respectively. The channel is assumed memoryless and without feedback, thus probabilities on n-vectors are given by:

\begin{equation}
P_{Y,Z|X,S}(y^n,z^n|x^n,s^n) =\prod_{i=1}^n P_{Y,Z|X,S}(y_i,z_i|x_i,s_i).
\end{equation}
In addition, the channel probability function can be decomposed as $P_{Y,Z|X,S}(y_i,z_i|x_{i},s_i)=P_{Y|X,S}(y_i|x_{i},s_i)P_{Z|Y}(z_i|y_i)$, i.e, $(X,S)-Y-Z$ form a Markov chain.
Due to the Markov property of physically degraded BCs, only a cooperation link from Decoder Y to the Decoder Z will contribute to increasing the rate.

\begin{definition}\label{PDSDCode1}
A $((2^{nR_Z},2^{nR_Y}),2^{nC_{12}},n)$ code for the PDSD BC with noncausal side information available at the encoder and strong decoder (where switch $A$ is closed in Fig. \ref{channel1}) consists of two sets of integers, $\mathcal{M}_Z = \{1,2,...,2^{nR_Z}\}$ and $\mathcal{M}_Y = \{1,2,...,2^{nR_Y}\}$, called message sets, an index set for the conference message $\mathcal{M}_{12} = \{1,2,...,2^{nC_{12}}\}$, an encoding function
\begin{equation}
f:\mathcal{M}_Z\times\mathcal{M}_Y\times\mathcal{S}^n\rightarrow \mathcal{X}^n\label{nocausalEncoder},
\end{equation}
a conference mapping
\begin{equation}
h_{12}:\mathcal{Y}^n\times\mathcal{S}^n\rightarrow \mathcal{M}_{12},
\end{equation}
and two decoding functions
\begin{eqnarray}
g_y&:&\mathcal{Y}^n\times\mathcal{S}^n\rightarrow \hat{\mathcal{M}_Y}\label{nocausalDecoder}\\
g_z&:&\mathcal{Z}^n\times\mathcal{M}_{12}\rightarrow \hat{\mathcal{M}_Z}.
\end{eqnarray}
\end{definition}

\begin{definition}\label{PDSDCodeCausal}
The definition of an $((2^{nR_Z},2^{nR_Y}),2^{nC_{12}},n)$ code for the PDSD BC with causal side information and noninformed decoders (where switch $A$ is open and switch $B$ is closed in Fig. \ref{channel1}) follows Definition \ref{PDSDCode1} above, except that the encoder (\ref{nocausalEncoder}) is replaced by a sequence of encoders:
\begin{equation}
f_i:\mathcal{M}_Z\times\mathcal{M}_Y\times\mathcal{S}^i\rightarrow \mathcal{X}^n, \ \ \  i=1,2,\ldots, n
\end{equation}
and Decoder Y's decoding function (\ref{nocausalDecoder}) is replaced by
\begin{eqnarray}
g_y&:&\mathcal{Y}^n\rightarrow \hat{\mathcal{M}_Y}.
\end{eqnarray}
\end{definition}

\par In Definition \ref{PDSDCode1} there is no restriction on when the link $C_{12}$ can be used. However, if we restrict the link to be used before the sequence $Y^n$ is given to the strong decoder, then only information on the state sequence, $S^n$, can be sent there. This is the subject of the next definition.
\begin{definition}\label{PDSDCode3}
A $((2^{nR_Z},2^{nR_Y}),2^{nC_{12}},n)$ code for the PDSD BC with rate-limited side information at the weak decoder, illustrated in Fig. \ref{channel2},  consists of three sets of integers, $\mathcal{M}_Z = \{1,2,...,2^{nR_Z}\}$ and $\mathcal{M}_Y = \{1,2,...,2^{nR_Y}\}$, called message sets, an index set $\mathcal{M}_{12} = \{1,2,...,2^{nC_{12}}\}$, a channel encoding function
\begin{equation}
f:\mathcal{M}_Z\times\mathcal{M}_Y\times\mathcal{S}^n\rightarrow \mathcal{X}^n,
\end{equation}
a state encoding function
\begin{equation}
h_s:\mathcal{S}^n\rightarrow \mathcal{M}_{12},
\end{equation}
and two decoding functions
\begin{eqnarray}
g_y&:&\mathcal{Y}^n\times\mathcal{S}^n\rightarrow \hat{\mathcal{M}_Y}\\
g_z&:&\mathcal{Z}^n\times\mathcal{M}_{12}\rightarrow \hat{\mathcal{M}_Z}.
\end{eqnarray}
\end{definition}

The next definitions, which deal with probability of error, achievable rates, and capacity region, hold for all the three problems defined in definitions \ref{PDSDCode1} to \ref{PDSDCode3} above, with respect to the corresponding code definitions.
\begin{definition}
We define the average probability of error for a $((2^{nR_Z},2^{nR_Y}),2^{nC_{12}},n)$ code as follows:
\begin{equation}
P_e^{(n)} = Pr\big{(}(\hat{M}_Y\neq M_Y)\cup (\hat{M}_Z\neq M_Z)\big{)}.
\end{equation}
The average probability of error at each receiver is defined as
\begin{eqnarray}
P_{e,y}^{(n)} &=& Pr(\hat{M}_Y\neq M_Y)\\
P_{e,z}^{(n)} &=& Pr(\hat{M}_Z\neq M_Z).
\end{eqnarray}
\end{definition}
As is commonly held when discussing BCs, the average probability $P_e^{(n)}$ tends to zero as $n$ approaches infinity, iff $P_{e,y}^{(n)}$ and $P_{e,z}^{(n)}$ both tend to zero as $n$ approaches infinity.
\begin{definition}
A rate triplet $(R_Z,R_Y,C_{12})$ is achievable if there exists a sequence of codes $((2^{nR_Z},2^{nR_Y}),2^{nC_{12}},n)$ such that $P^{(n)}_e \rightarrow 0$ as $n\rightarrow \infty$.
\end{definition}

\begin{definition}
The capacity region is the closure of all achievable rates.
\end{definition}

\begin{figure}[h!]
    \begin{center}
        \begin{psfrags}
            \psfragscanon
            \psfrag{X}[][][0.9]{$S^n$}
            \psfrag{I}[][][0.9]{\ \ \ \ $M_Z$}
            \psfrag{Z}[][][0.9]{\ \ \ \ $M_Y$}
            \psfrag{J}[][][0.9]{\ \ \ \ \ \ \ \ \  Encoder}
            \psfrag{K}[][][0.9]{\ \ \ $X^n$}
            \psfrag{S}[][][0.9]{\ \ \ \ \ \ \ \ \ \ \ \ \ \ \ \ \ \ \  $P_{Y,Z|X,S}=$}
            \psfrag{L}[][][0.9]{\ \ \ \ \ \ \ \ \ \ \ \ \ \ \ \ \ \ \ \ $P_{Y|X,S}P_{Z|Y}$}
            \psfrag{M}[][][0.9]{\ \ \ \  $Y^n$}
            \psfrag{N}[][][0.9]{\ \ \ \ $Z^n$}
            \psfrag{O}[][][0.9]{\ \ \ \ \ \ \ \ \  Decoder Y}
            \psfrag{P}[][][0.9]{\ \ \ \ \ \ \ \ \  Decoder Z}
            \psfrag{A}[][][0.9]{\ \ \ \ \ \ Rate $\leq C_{12}$}
            \psfrag{Q}[][][0.9]{\ \  $\hat{M}_Y$}
            \psfrag{R}[][][0.9]{\ \  $\hat{M}_Z$}
            \psfrag{B}[][][0.9]{\ \ \ \ \ \ State}
            \psfrag{T}[][][0.9]{\ \ \ \ \ \ \ \ Encoder}
            \psfrag{C}[][][0.9]{ $S^n$}
            \centerline{\includegraphics[scale = .6]{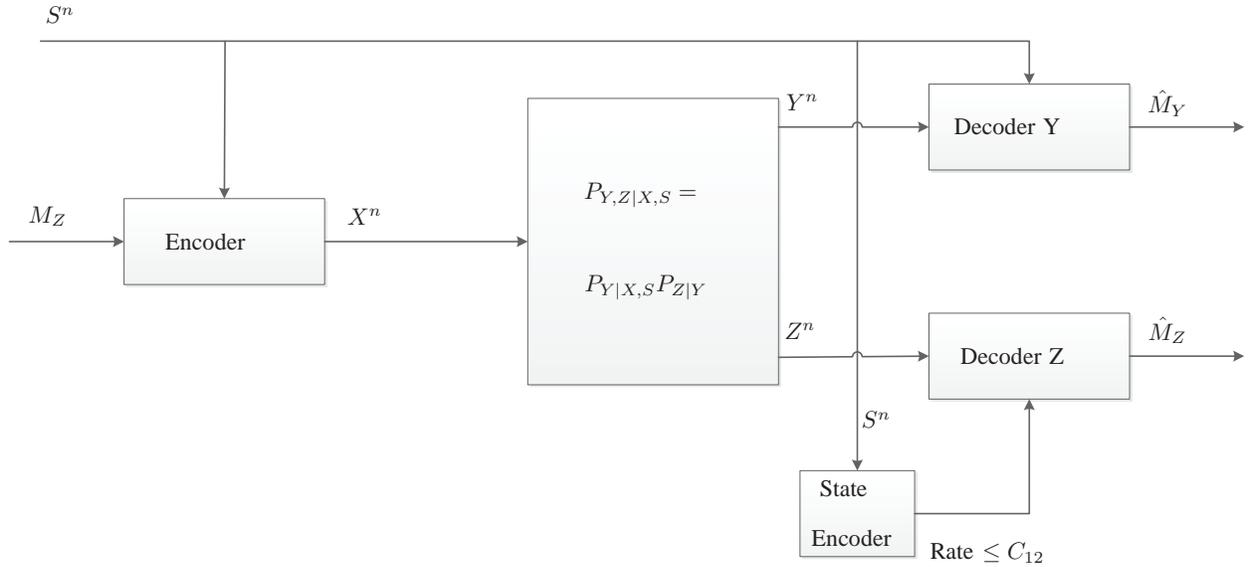}}
            \caption{The physically degraded, state-dependent BC with full state information at the encoder and one decoder together with rate-limited state information at the other decoder. This model describes the case where the cooperation between the decoders is confined such that it takes place prior to the decoding of the messages. Therefore, the only information that the strong decoder, Decoder Y, can send to the weaker decoder, Decoder Z, is regarding the state sequence. However, the state is only partially available at Decoder Z, since it is sent rate-limited due to the limited capacity of the link between the decoders.} \label{channel2}
            \psfragscanoff
        \end{psfrags}
     \end{center}
     \vspace{-4mm}
 \end{figure}


\section{Main Results and Insights}\label{SectionMainResults}
\subsection{Capacity Region of the PDSD BC with Cooperating Decoders}\label{SectionCapacity1}
\par We begin by stating the capacity region for the PDSD BC illustrated in Fig. \ref{channel1} (switch $A$ is closed) in the following theorem.
\begin{theorem} \label{Theorem1}
The capacity region of the PDSD BC, $(X,S)-Y-Z$, with noncausal state information known at the encoder and at Decoder Y, with cooperating decoders, is the closure of the set that contains all the rates $(R_Z,R_Y)$ that satisfy
\begin{subequations}\label{CapacityRegion1}
\begin{eqnarray}
R_Z&\leq& I(U;Z)-I(U;S)+C_{12}\label{CapacityRegion1.1}\\
R_Y&\leq& I(X;Y|U,S)\label{CapacityRegion1.2}\\
R_Z+R_Y&\leq& I(X;Y|S)\label{CapacityRegion1.3},
\end{eqnarray}
for some joint probability distribution of the form
\begin{equation}
P_{S,U,X,Z,Y} = P_{S}P_{U|S}P_{X|S,U}P_{Y|X,S}P_{Z|Y}.\label{distribution1}
\end{equation}
\end{subequations}
\end{theorem}

\par The proof is given in Sections \ref{Achievability1} and \ref{Converse}. The achievability of Theorem \ref{Theorem1} is proved by using techniques that include triple-binning and superposition coding. The main idea is to identify how to best use the capacity link between the decoders. In particular, we want maximize the potential use of the capacity link while simultaneously successfully balancing the allocation of rate resources between the two messages. Changes in the allocation of resources between the messages $M_Y$ and $M_Z$ are possible due to the fact that we use superposition coding. Using this coding method, Decoder Y, which is the strong decoder, also decodes the message $M_Z$ intended for Decoder Z. This allows us to shift resources between the messages and thus increase the rate resources of $M_Z$ at the expense of the message $M_Y$. Decoder Y can then send information about $M_Z$ to Decoder Z by using the capacity link between them. Therefore, the optimal coding scheme balances the distribution of rate resources between the messages, taking into account that additional information can be sent through the capacity link.

\par The additional information sent from Decoder Y to Decoder Z comes into play via the use of binning, that is, we divide the messages $M_Z$ among superbins in ascending order. Next, we use a Gel'fand-Pinsker code for each superbin. Now, we can redirect some of the rate resources of the message $M_Y$ to send the superbin index that contains $M_Z$. Decoder Y, which decodes both messages, sends this superbin index to Decoder Z through the capacity link between the decoders. Decoder Z then searches for $M_Z$ only in that superbin, by using joint typicality methods. By utilizing the capacity link through adding the superbining measure, we can increase the rate of $M_ Z$ achieved using the standard Gel'fand-Pinsker coding scheme.

\par Nevertheless, if the capacity link between the decoders is very large, there is still a restriction on the amount of information we can send through it. This restriction is reflected through the bound on the rate sum, $R_Z+R_Y\leq I(X;Y|S)$. This bound indicates that we cannot send more information about $(M_Y,M_Z)$ through this setting compared to the information we could have sent about $(M_Y,M_Z)$ through a state dependent point-to-point channel, where the state information is known at both the encoder and decoder. Moreover, note that
\begin{equation}
R_Z+R_Y\leq I(X;Y|S)=I(U,X;Y|S)=I(U;Y|S)+I(X;Y|U,S),
\end{equation}
that is, if we have a large capacity link between the decoders, we have a tradeoff between sending information about $M_Z$ and $M_Y$. If we choose to send information about $M_Y$ at the maximal rate possible, $I(X;Y|U,S)$, then the maximal rate we can send $M_Z$ is $I(U;Y|S)$. In contrast, we can increase the rate of $M_Z$ (up to the minimum between $I(U;Z)-I(U;S)+C_{12}$ or $I(X;Y |S)$) at the expense of reducing the rate of $M_Y$. For example, if we have an infinite capacity link between the decoders, we can consolidate the two decoders into one decoder which will to decode both messages.


\par Another interesting insight is revealed by comparing a special case of the result presented in Theorem \ref{Theorem1} with the result presented in \cite{R.Dabora_S.ServettoRS} which is referred to as the physically degraded BC with cooperating receivers (i.e., the BC model where the channel is not state dependent).
\par Let us consider the special case of Theorem \ref{Theorem1} where $S=\emptyset$. As a result, the region (\ref{CapacityRegion1}) reduces to
\begin{subequations}\label{noState}
\begin{eqnarray}
R_Z&\leq& I(U;Z)+C_{12}\label{noState1.1}\\
R_Y&\leq& I(X;Y|U)\label{noState1.2}\\
R_Z+R_Y&\leq& I(X;Y)\label{noState1.3},
\end{eqnarray}
\end{subequations}
for some
\begin{equation}
P_{U,X,Z,Y} =P_{X,U}P_{Y|X}P_{Z|Y}.\nonumber
\end{equation}
This special case was studied in \cite{R.Dabora_S.ServettoRS}, where a different expression for the capacity region was found:
\begin{subequations}\label{DaboraRegion}
\begin{eqnarray}
R_Z&\leq& \min\{I(U;Z)+C_{12} \ , \ I(U;Y)\}\label{DaboraRegion1.1}\\
R_Y&\leq& I(X;Y|U)\label{DaboraRegion1.2}
\end{eqnarray}
\end{subequations}
for some
\begin{equation}
P_{U,X,Z,Y} =P_{X,U}P_{Y|X}P_{Z|Y}.\nonumber
\end{equation}

Since the two regions, (\ref{noState}) and (\ref{DaboraRegion}), are shown to be the capacity regions of the same setting, it indicates that the two regions should be equivalent. It is simple to show that (\ref{DaboraRegion})$\subseteq$(\ref{noState}). However, the reverse inclusion is not so straightforward. This raises the question: are the regions indeed equal? The answer is given in the following corollary.
\begin{corollary}\label{cor: eq.regions}
The two regions, (\ref{noState}) and (\ref{DaboraRegion}), are equivalent.
\end{corollary}

 A direct proof, which is not based on the fact that both regions characterize the same capacity, is given in Section \ref{proof:cor:eq_regions}. The main idea is to find a specific choice of $U$, for any rate pair $(R_Z,R_Y)$ in (\ref{noState}), such that $(R_Z,R_Y)$ satisfy (\ref{DaboraRegion}), implying that
(\ref{noState})$\subseteq$(\ref{DaboraRegion}). Hence, we conclude that (\ref{noState}) and (\ref{DaboraRegion}) are equivalent.


\subsection{Causal Side Information}\label{SectionCapacity2}
 \par Consider the case where the state is known to the encoder in a causal manner, i.e., at each time index, $i$, the encoder has access to the sequence $s^i$. This setting is illustrated in Fig. \ref{channel1}, where we take switch $A$ to be open and switch $B$ to be closed. In this scenario, the encoder is the only user with information regarding the state sequence, in contrast to the noncausal case, where the strong decoder also has access to the channel states. The capacity region for this setting is characterized in the following theorem.
 \begin{theorem}\label{Theorem2}
The capacity region for the PDSD BC with cooperating decoders and causal side information known at the encoder is the closure of the set that contains all the rates $(R_Z,R_Y)$  that satisfy
\begin{subequations}\label{CapacityRegion2}
\begin{eqnarray}
R_Z&\leq& I(U;Z)+C_{12}\label{CapacityRegion2.1}\\
R_Y&\leq& I(V;Y|U)\label{CapacityRegion2.2}\\
R_Z+R_Y&\leq& I(V,U;Y)\label{CapacityRegion2.3},
\end{eqnarray}
for some joint probability distribution of the form
\begin{eqnarray}
P_{S,U,V,X,Z,Y} &=& P_{S}P_{U,V}P_{X|S,U,V}P_{Y|X,S}P_{Z|Y}.\label{distribution2}
\end{eqnarray}
\end{subequations}
\end{theorem}
The proof is given in Section \ref{SectionCausalProof}.

\subsection{Capacity Region of the PDSD BC with Rate-Limited State Information at the Weak Decoder}\label{SectionCapacity3}
In the previous two cases there was no restriction on when the cooperation link $C_{12}$ is to be used. In the following case, we consider a setting where the cooperation is restricted to being used before the outputs $Y^n$ are given to the strong decoder. Therefore, we can only use the cooperation link to send the weak decoder information about the state sequence, $s^n$, from the strong decoder. Moreover, since there is a limit on the information we can send through the link, the weak decoder receives rate-limited information about the channel states. Hence, we model this setting as a PDBC with noncausal state information at the encoder and strong decoder and rate-limited state information at the weak decoder. We state the capacity region for this setting in the following theorem
\begin{theorem} \label{Theorem3}
The capacity region of the PDSD BC, $(X,S)-Y-Z$, with rate-limited state information at the weak decoder, illustrated in Fig. \ref{channel2}, is the closure of the set that contains all the rates $(R_Z,R_Y)$ that satisfy
\begin{subequations}\label{CapacityRegion3}
\begin{eqnarray}
R_Z&\leq& I(U;Z,S_d)-I(U;S,S_d)\label{CapacityRegion3.1}\\
R_Y&\leq& I(X;Y|U,S,S_d)\label{CapacityRegion3.2}
\end{eqnarray}
for some joint probability distribution of the form
\begin{equation}
P_{S,U,X,Z,Y} = P_{S}P_{S_d,U,X|S}P_{Y|X,S}P_{Z|Y}\label{distribution3}
\end{equation}
such that
\begin{equation}
C_{12}\geq I(S;S_d)-I(Z;S_d).\label{constraint1}
\end{equation}
\end{subequations}
\end{theorem}

\begin{remark}
As was noted in \cite{Steinberg08RateLimitedAtDecoder}, we can replace the rate bound on $R_Z$ in (\ref{CapacityRegion3.1}) with the bound:
\begin{eqnarray}
R_Z&\leq& I(U;Z|S_d)-I(U;S|S_d).
\end{eqnarray}
This can easily be seen by applying the chain rule on the expressions on the right-hand side of (\ref{CapacityRegion3.1}) as follows
\begin{align*}
I(U;Z,S_d)-I(U;S,S_d)&=I(U;S_d)+I(U;Z|S_d)-I(U;S_d)-I(U;S|S_d)\\
&=I(U;Z|S_d)-I(U;S|S_d).
\end{align*}
\end{remark}

\begin{remark}
Observe that the region (\ref{CapacityRegion3}) is contained in the region (\ref{CapacityRegion1}) given in Theorem \ref{Theorem1}. The rate $R_Z$ can be further bounded as follows:
\begin{align*}
R_Z &\leq I(U;Z|S_d)-I(U;S|S_d)\\
    &= I(U,S_d;Z)-I(U,S_d;S)- I(S_d;Z)+I(S_d;S)\\
    &\leq I(U,S_d;Z)-I(U,S_d;S)+C_{12}\\
    &= I(\tilde{U};Z)-I(\tilde{U};S)+C_{12},
\end{align*}
where we take $\tilde{U} = (U,S_d)$. Furthermore, with the definition of $\tilde{U}$, the rate $R_Y$ is bounded by
\begin{align*}
R_Y &\leq I(X;Y|U,S_d,S)\\
    &= I(X;Y|\tilde{U},S).
\end{align*}
Finally, notice that  the sum of rates, $R_Z+R_Y$, satisfies $R_Z+R_Y\leq I(X;Y|S)$, since
\begin{align*}
R_Z+R_Y &\leq I(X;Y|S)\\
        &= I(\tilde{U},X;Y|S)\\
        &= I(U,S_d:Y|S) + I(X;Y|U,S_d,S)\\
        &= I(U,S_d:Y,S) -I(U,S_d:S) + I(X;Y|U,S_d,S),
\end{align*}
where
\begin{align*}
R_Z &\leq I(U;Z|S_d)-I(U;S|S_d)\\
    &\leq I(U,S_d:Y,S) -I(U,S_d:S).
\end{align*}
Hence we have that (\ref{CapacityRegion3}) $\subseteq$ (\ref{CapacityRegion1}).
\end{remark}
\par The proof of Theorem \ref{Theorem3} is given in Sections \ref{Achievability2} and \ref{Converse2}. The coding scheme that achieves the capacity region (\ref{CapacityRegion3}) uses techniques that include superposition coding and Gel'fand-Pinsker coding, which are similar to the proof of Theorem \ref{Theorem1}, with an addition of  Wyner-Ziv compression. The main idea is based on a Gel'fand-Pinsker code, but with several extensions. First, note that the Channel Encoder and Decoder Y are both informed of the sequence $s_d^n$, the compressed codeword that the State Encoder sends Decoder Z. This is due to the fact that they both know the state sequence $s^n$ and the State Encoder's strategy. The encoder then uses this knowledge to find a codeword $u^n$, in a bin associated with $m_Z$, that is jointly typical not only with $s^n $, as in the original Gel'fand-Pinsker scheme, but also with $s_d^n$. Each codeword $u^n$ is set as the center of $2^{nR_Y}$ bins, one bin for each message $m_Y$. Once a codeword $u^n$ is chosen, we look in its satellite bin $m_Y$ for a codeword $x^n$ such that it is jointly typical with that $u^n $, the state sequence $s^n$ and with $s^n_d$. Upon transmission, the State Encoder chooses a codeword $s_d^n$, the Channel Encoder chooses a codeword $u^n$ and a corresponding $x^n$, where $x^n$ is transmitted through the channel. Consequently, identifying $x^n$ leads to the identification of $u^n$.

\par At the decoder's end, the State Encoder sends Decoder Z a compressed version of $s^n$ by using the codewords $s_d^n$ in a Wyner-Ziv scheme, where Decoder Z uses the channel output $z^n$ as side information. The joint typicality of $z^n$ and $s_d^n$ used for decoding is a result of the channel Markov $S_d-(X,S)-Y-Z$ and the fact that the codewords $x^n$ are generated $\sim \prod_{i=1}^n p(x_i|u_i,s_{d,i})$. Finally, $s_d^ n$ is used as side information to identify the codeword $u^n$. As for the strong decoder, Decoder Y looks for codewords $u^n$ and $x^n$ that are jointly typical with the received $y^n$, the state sequence $s^n$, and the codeword $s_d^n$.

\section{Is Rate Splitting Optimal for Cooperation?}\label{SectionRS}

\par When dealing with cooperation settings, the most common approach is the use of rate splitting. Many coding schemes based on rate splitting have been known to achieve the capacity of channels involving cooperation. For example, rate splitting is the preferred coding method when coding for cooperative MACs, and it has been shown to be optimal \cite{Willems83_cooperating}, \cite{Willems85_cribbing_encoders} . However, when dealing with cooperative BCs, we show, by a numerical example, that rate splitting schemes are not necessarily optimal. Moreover, other techniques, such as binning, strictly outperform rate splitting.

\par The main idea of the rate splitting scheme is to split the message intended for the weaker decoder, $M_Z$, into two messages, $M_Z=(M_{Z_1},M_{Z_2})$. Next, we reorganize the messages. We concatenate part of the message intended for the weak decoder, $M_{Z_2}$, to the message intended for the strong decoder, $ M_Y$. In addition, we define new message sets, ${M'}_Z=M_{Z_1}$ and ${M'}_Y=(M_Y,M_{Z_2})$, where we choose $M_{Z_2}$ to be of size $\leq C_{12}$. Now that we have a new message set $({M'}_Z,{M'}_Y)$ , we transmit by using a Gelfand-Pinsker superposition coding scheme, such as the one described in \cite{steinberg_causal_BC_2005}. Once the strong decoder decodes both messages, $({M'}_Z,{M'}_Y) $ (which, in whole, equal $({M}_Z,{M}_Y)$), it uses the capacity link between the decoders, $C_{12 }$, to send the message $M_{Z_2}$ to the weak decoder. To sum up, this scheme results with the strong decoder decoding both messages, and the weak decoder decoding the original message $M_Z$.

\par The achievability scheme that uses the rate splitting method closely follows the achievability of Theorem \ref{Theorem1}, but with some alterations. In the rate splitting scheme, we define two sets of messages $\mathcal{M}_{Z_1}=\{1,2,...,2^{nR_{Z_1}}\}$ and $\mathcal{M}_{Z_2}=\{1,2,...,2^{nR_{Z_2}}\}$, where $|\mathcal{M}_{Z_1}||\mathcal{M}_{Z_2}|=|\mathcal{M}_Z|$ and $R_Z=R_{Z_1}+R_{Z_2}$, such that each message, $m_{Z}\in\{1,2,...,2^{n{R}_Z}\}$, is uniquely defined by a pair of messages $(m_{Z_1},m_{Z_2})$. Using these definitions, we can define a new pair of messages, $({m}'_Z,{m}'_Y)$, where we take $m'_Z=m_{Z_1}$, $R'_Z=R_{Z_1}$, $m'_y=(m_Y,m_{Z_2})$ and $R'_Y=R_Y+R_{Z_2}$. The code is now constructed in a similar manner to the code described in the triple-binning achievability scheme with respect to $(m'_Z,m'_Y)$, despite the fact that in this scheme additional partitioning into superbins is not required. However, we will see that this fact turns out to be significant.

\par To transmit $(m_Y,m_Z)$ in the encoding stage, we first construct the corresponding pair $(m'_{Z},m'_{Y})$. The rest of the encoding is preformed in a manner similar to the encoding in Section \ref{Achievability1} with respect to the constructed $({m'}_Z,{m'}_Y)$. The decoding stage is also similar, except that now Decoder Y, upon decoding the messages $(\hat{m'}_Y,\hat{m'}_Z)$, uses the link $C_{12}$ to send the message $\hat{M}_{Z_2}$ to Decoder Z (instead of a bin number as in the achievability of Theorem \ref{Theorem1}). The code construction is illustrated in Fig. \ref{AchievadilityFigRateSplitting}.

\begin{figure}[h!]
\begin{center}
\begin{psfrags}
    \psfragscanon
    \psfrag{A}[][][0.9]{\ \ \ \ \ a codeword $u^n(m'_Z)$}
    \psfrag{B}[][][0.7]{$1$}
    \psfrag{C}[][][0.7]{$2$}
    \psfrag{D}[][][0.7]{\ \ \ \ \ \ \ \ \ \ \ \ \ \ \ \ \ $2^{n(I(U;Z)-I(U;S)}$}
    \psfrag{E}[][][0.5]{}
    \psfrag{F}[][][0.9]{\ \ \ \ \ \ \ \ \ \ \ \ \ \ \ \ \ \ \ \ \ \ \ \ \ \ \ \ \ \ \ $2^{nI(U;S)}$
 codewords in a bin}
    \psfrag{I}[][][0.9]{\ \ \ \ \ \ \ \ \ \ \ \ \ \ \ \ \ \ \ \ \ \ \ \ \ \ \ \ \ \ \ $2^{nI(U;Z)}$
 codewords in total}
    \psfrag{K}[][][0.9]{\ \ \ \ \ \ \ \ \ \ \ \ \ \ \ \ \ \ \ \ \ $u^n(m'_Z)$ cloud center}
    \psfrag{L}[][][0.9]{\ \ \ \ \ \ \ \ \ \ \ \ \ \ \ \ \ \ \ \ \ \ \ \ \ \ \ \ \ \ \ \ \ \ \ \ \
\ \ \ \ \ \ \ \ \ \ \ \ \ \ \ \ \ \ \ \ \ \ \ \ \ \ \ \ \ \ $2^{nI(X;Y|U,S)}$ satellite bins.
Each bin contains $2^{nI(X;S|U)}$ codewords $x^n(m'_Z,m'_Y)$}
    \centerline{\includegraphics[scale = 1]{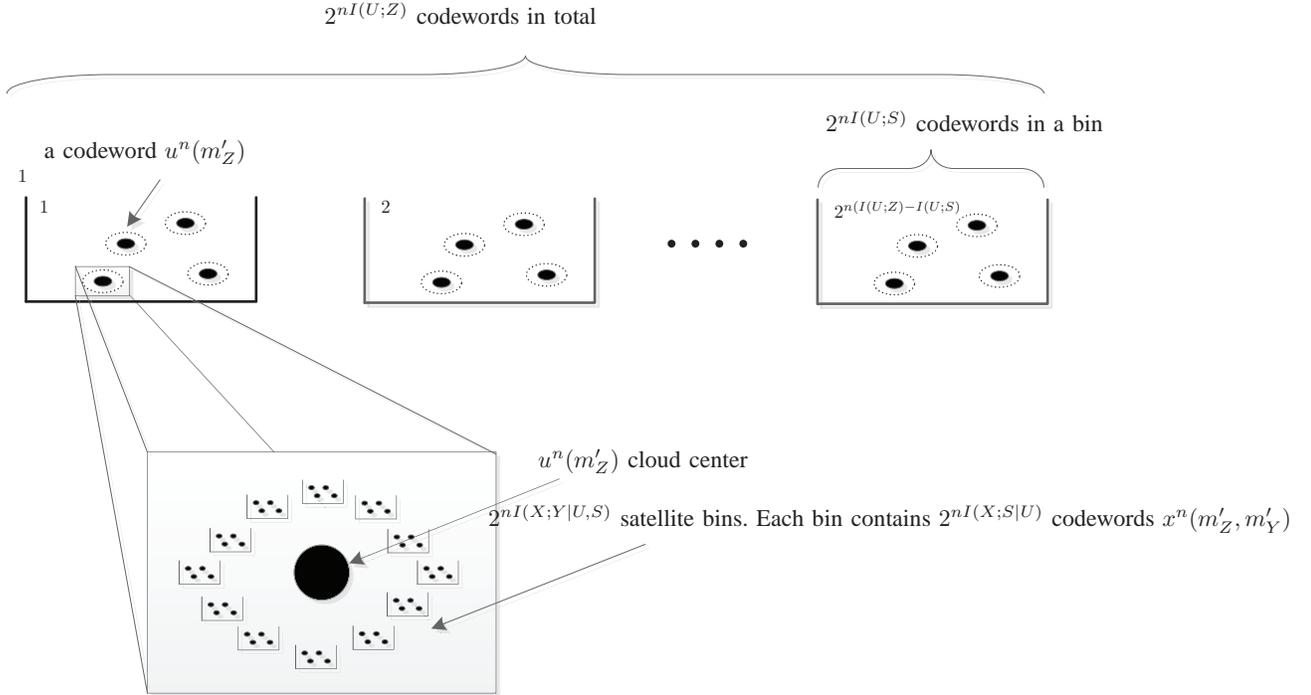}}
    \caption{The code construction for the rate splitting scheme. We can see that the code
construction is similar to the triple-binning scheme, except that here we do not partition the
bins associated with the messages $\tilde{m}_Z$ into superbins.} \label{AchievadilityFigRateSplitting}
\end{psfrags}
\end{center}
\end{figure}

Using the achievability result of the capacity region found in \cite[Theorem 3]{steinberg_causal_BC _2005} for the PDSD BC with state information known at the encoder and decoder, together with the fact that the rate $R_{Z_2}$ cannot be negative or greater than $C_{12}$, we derive the following bounds:
\begin{subequations}
\begin{eqnarray}
R_{Z_2}&\leq&C_{12}\label{RateSplitting1.1}\\
R_{Z_2}&\geq& 0\label{RateSplitting1.2}\\
R_{Z_1}&\leq& I(U;Z)-I(U;S)\label{RateSplitting1.3}\\
R_Y+R_{Z_2}&\leq& I(X;Y|U,S)\label{RateSplitting1.4}\\
R_{Z_1}+R_{Z_2}+R_Y&\leq& I(X;Y|S)\label{RateSplitting1.5}.
\end{eqnarray}
\end{subequations}

Recalling that $R_Z=R_{Z_1}+R_{Z_2}$, we substitute $R_{Z_1}$ with $R_Z-R_{Z_2}$ in the bounds (\ref{RateSplitting1.3}) and (\ref{RateSplitting1.5}). Next, by using the Fourier-Motzkin elimination, we eliminate the bounds that contain $R_{Z_2}$, (\ref{RateSplitting1.1}), (\ref{RateSplitting1.2}) and (\ref{RateSplitting1.4}). The resulting region is the following bounds on $R_Z$ and $R_Y$
\begin{subequations}\label{RateSplitting2}
\begin{eqnarray}
R_{Z}&\leq& I(U;Z)-I(U;S)+C_{12}\label{RateSplitting2.2}\\
R_Y&\leq& I(X;Y|U,S)\label{RateSplitting2.3}\\
R_Z+R_Y&\leq& I(U;Z)-I(U;S)+I(X;Y|U,S)\label{RateSplitting2.4}.
\end{eqnarray}
\end{subequations}
This region, (\ref{RateSplitting2}), is the achievable region as a result of rate splitting.
\par We note that in the process we derive an additional bound on the rate sum $R_Z+R_Y\leq I(X;Y|S)$; however, we can see that this bound satisfied automatically by satisfying (\ref{RateSplitting2.4}), since
\begin{align}
I(X;Y|S)&= I(U,X;Y|S)\nonumber\\
        &= I(U;Y|S)+I(X;Y|U,S)\nonumber\\
        &= I(U;Y,S)-I(U,S)+I(X;Y|U,S)\nonumber\\
        &\geq I(U;Z)-I(U;S)+I(X;Y|U,S)\label{RateSumBound},
\end{align}
where the last inequality is due to the degradedness properties of the channel. Moreover, we also omit the bound $C_{12}\geq0$, which is follows from the problem setting.
\par Examining the region (\ref{RateSplitting2}) we notice that its form differs from the capacity region of this channel (\ref{CapacityRegion1}). Therefore, an interesting question rises: Are rate splitting coding schemes optimal for BCs with cooperating decoders? We answer this question in the following lemma.
\begin{lemma}\label{lemma1RS}
Using rate splitting coding for BCs with cooperating decoders is not necessarily optimal.
\end{lemma}

\begin{proof}
We would like to show that the rate splitting coding scheme that derives the region (\ref{RateSplitting2}) does not achieve the capacity of the channel given by (\ref{CapacityRegion1}) in Theorem \ref{Theorem1}. In order to do so, we need to show that the region (\ref{CapacityRegion1}) is strictly larger than the region achievable by rate splitting, (\ref{RateSplitting2}). The region (\ref{CapacityRegion1}) is shown to be achievable in Section \ref{Achievability1} by using triple-binning. Thus, by showing that (\ref{CapacityRegion1}) is strictly larger than (\ref{RateSplitting2}), we can conclude that the rate splitting method is not optimal.

\par Firstly, it is easy to see that the region (\ref{CapacityRegion1}) contains (\ref{RateSplitting2}), since the bounds on $R_Z$ and $R_Y$ are the same, yet the bound on the rate sum (\ref{CapacityRegion1.3}) is greater than or equal to (\ref{RateSplitting2.4}), as shown in (\ref{RateSumBound}).
However, to show that (\ref{CapacityRegion1}) is strictly larger than (\ref{RateSplitting2}), we need to show that for all distributions of the form (\ref{distribution1}):
 \begin{equation}
 \Big{\{}\exists (R_Z,R_Y)\in (\ref{CapacityRegion1}): (R_Z,R_Y)\notin (\ref{RateSplitting2})\Big{\}}.
 \end{equation}
This is not an easy task. If we look at the regions in their general form, we need to find a pair $(R_Z,R_Y)\in (\ref{CapacityRegion1})$ and show that for every random variable $U$ we choose, $(R_Z,R_Y)\notin (\ref{RateSplitting2})$. Nevertheless, we can show that (\ref{CapacityRegion1}) is strictly larger than (\ref{RateSplitting2}) by considering a specific channel and and showing that for this specific setting $(\ref{RateSplitting2})\subset(\ref{CapacityRegion1})$.

\subsection{The Special Case of the Binary Symmetric Broadcast Channel}\label{SubSectionBinarySymetric}

\par Consider the binary symmetric BC, \cite[Section 5.3]{ElGammalKim10LectureNotes}, illustrated in Fig. \ref{BinarySymmetricBC}. Here, $Y=X\oplus W_1$, $Z=X\oplus W_2$, where $W_1 \sim$Ber($p_1$) and $W_2\sim$Ber($p_2$). Note that we can present this channel as a physically degraded BC, where $Y=X\oplus {W}_1$, $Z=X\oplus \tilde{W}_2$ and ${W}_1\sim$Ber($p_1$), $\tilde{W}_2\sim$Ber($\frac{p_2-p_1}{1-2p_2}$). This channel is a special case of our setting, where the channel is not state-dependent (hence, we take the state as a constant).

\begin{figure}[h!]
\begin{center}
\begin{psfrags}
    \psfragscanon
    \psfrag{A}[][][0.9]{\ \ \ \ \ $X$}
    \psfrag{B}[][][0.9]{\ \ \ \ \ \ \ \ \ \ \ \ \ \ \ \ \ \ \ \ $W_1\sim$Ber($p_1$)}
    \psfrag{C}[][][0.9]{\ \ \ \ \ \ \ \ \ \ \ \ \ \ \ \ \ \ \ \ $W_2\sim$Ber($p_2$)}
    \psfrag{D}[][][0.9]{\ \ \  $Y$}
    \psfrag{E}[][][0.9]{\ \ \ $Z$}
    \psfrag{F}[][][0.9]{\ \ \ \ \ \ \ \ \ \ \ \ \ \ \ \ \ \ \ \ $W_1\sim$Ber($p_1$)}
    \psfrag{G}[][][0.9]{\ \ \ \ \ \ \ \ \ \ \ \ \ \ \ \ \ \ \ \ $\tilde{W}_2\sim$Ber($\frac{p_2-p_1}{1-2p_2}$)}
    \centerline{\includegraphics[scale = 1.5]{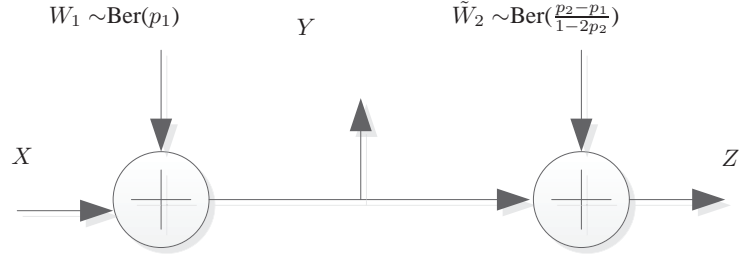}}
    \caption{The physically degraded binary symmetric BC.} \label{BinarySymmetricBC}
\end{psfrags}
\end{center}
\end{figure}

\par Following closely the arguments given in \cite[Section 5.4.2]{ElGammalKim10LectureNotes} we can upper bound the region (\ref{RateSplitting2}) by considering all the sets of rate pairs $(R_Z,R_Y)$ such that
\begin{subequations}\label{BinaryRS1}
\begin{eqnarray}
R_Z&\leq& 1-H(\alpha\ast p_2)+C_{12}\label{BinaryRS1.1}\\
R_Y&\leq& H(\alpha\ast p_1)- H(p_1)\label{BinaryRS1.2}\\
R_Z+R_Y&\leq& 1-H(\alpha\ast p_2)+H(\alpha\ast p_1)- H(p_1)\label{BinaryRS1.3}.
\end{eqnarray}
\end{subequations}
for some $\alpha\in[0,\frac{1}{2}]$. In contrast, we can show that by taking $X=U\oplus V$, where $U\sim$Ber($\frac{1}{2}$) and $V\sim$Ber($\alpha$), and calculating the corresponding expressions of (\ref{CapacityRegion1}), the following region of rate pairs, $(R_Z,R_Y)$ such that
\begin{subequations}\label{BinaryRS2}
\begin{eqnarray}
R_Z&\leq& 1-H(\alpha\ast p_2)+C_{12}\label{BinaryRS2.1}\\
R_Y&\leq& H(\alpha\ast p_1)- H(p_1)\label{BinaryRS2.2}\\
R_Z+R_Y&\leq& 1- H(p_1)\label{BinaryRS2.3},
\end{eqnarray}
\end{subequations}
is achievable via the binning scheme.

\begin{figure}[h!]
\begin{center}
\begin{psfrags}
    \psfragscanon
    \psfrag{A}[][][0.9]{$R_Y$}
    \psfrag{B}[][][0.9]{$R_Z$}
    \psfrag{D}[][][0.9]{Rate Region for $C_{12}=0.05$}
    \psfrag{E}[][][0.9]{Rate Region for $C_{12}=0.1$}
    \psfrag{F}[][][0.9]{Rate Region for $C_{12}=0.2$}
    \psfrag{G}[][][0.9]{Rate Region for $C_{12}=0.3$}
    \centerline{\includegraphics[scale = 0.7]{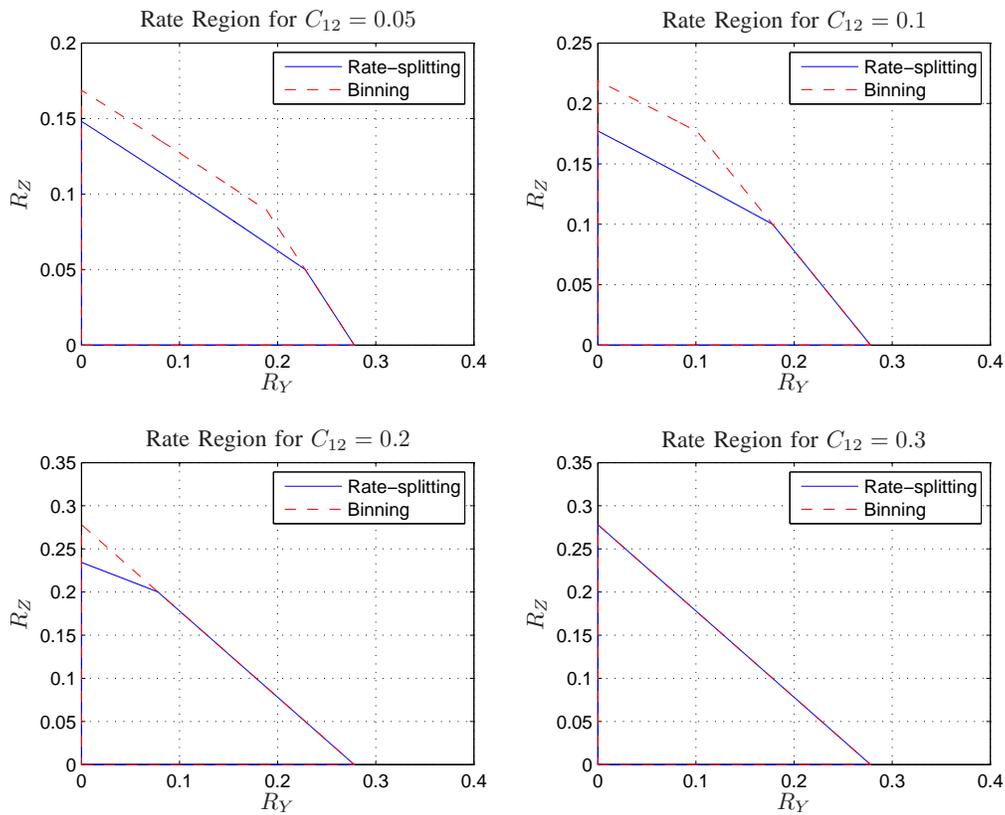}}
    \caption{The upper bound for the region (\ref{RateSplitting2}), which is calculated in (\ref{BinaryRS1}), is plotted using the solid line and corresponds to the smaller region. The achievable region for (\ref{CapacityRegion1}), given by the expressions in (\ref{BinaryRS2}), is plotted using the dashed line and corresponds to the larger region. Both regions in all the figures are plotted for values of $p_1=0.2$, $p_2=0.3$, where each figure corresponds to a different value of $C_{12}$.} \label{BinarySymmetricBCFig}
\end{psfrags}
\end{center}
\end{figure}

\par Consequently, we can see that the region (\ref{BinaryRS2}) is strictly larger than the region (\ref{BinaryRS1}). For example, consider Fig. \ref{BinarySymmetricBCFig}, where we take $p_1=0.2$, $p_2=0.3$ and $C_{12}=0.03$. Looking at both regions, we can see that taking by $R_Y$ to be zero, the point $(R_Z,R_Y)=(0.1487,0)$ is achievable in the binning region (\ref{BinaryRS2}) (the doted line) for $\alpha=0$, but it is not achievable in the rate splitting region (\ref{BinaryRS1}) (the solid line) for any value of $\alpha\in[0,\frac{1}{2}]$.


\par Thus, for the binary symmetric BC we have shown that an achievable region derived from (\ref{CapacityRegion1}) by a specific choice of $U$ is strictly larger than the upper bound for the region (\ref{RateSplitting2}). Therefore, we can conclude that (\ref{CapacityRegion1}) is strictly larger then (\ref{RateSplitting2}) and that the rate splitting coding scheme is not necessarily optimal for BCs.
\end{proof}


\section{Proofs}\label{SectionProofs}
\subsection{Proof of Corollary \ref{cor: eq.regions}}\label{proof:cor:eq_regions}
Let us denote our region without states by ${\cal A}$. It is characterized as the union of all rate pairs $(R_y,R_z)$ satisfying:
\begin{subequations}
\label{eq:YS1}
\begin{IEEEeqnarray}{rCl}
R_z &\leq& I(U;Z) +C_{12} \label{eq:YS1_1}\\
R_y &\leq& I(X;Y|U) \label{eq:YS1_2}\\
R_y+R_z &\leq& I(X;Y) \label{eq:YS1_3}
\end{IEEEeqnarray}
\end{subequations}
for some joint distribution
\begin{equation}
P_{U,X,Y,Z} = P_U P_{X|U} P_{Y|X}P_{Z|Y} \label{eq:YS_joint}
\end{equation}
where $P_{Y|X}P_{Z|Y}$ is the original BC (without states). The region of Dabora and Servetto, presented in \cite{R.Dabora_S.ServettoRS}, is the union of all rate pairs $(R_y,R_z)$ satisfying \begin{subequations}
\label{eq:YS2}
\begin{IEEEeqnarray}{rCl}
R_z &\leq& \min \left\{ I(U;Z) + C_{12}, I(U;Y) \right\} \label{eq:YS2_1}\\
R_y &\leq& I(X;Y|U) \label{eq:YS2_2}
\end{IEEEeqnarray}
\end{subequations}
for some joint distribution~\eqref{eq:YS_joint}. For brevity, we denote the region of Dabora and Servetto by ${\cal B}$. It is simple to show that ${\cal B}\subseteq {\cal A}$. We now proceed to show the reverse inclusion, i.e., ${\cal A}\subseteq {\cal B}$. Let $(R_y,R_z)$ be a rate pair in ${\cal A}$, achieved with a given pair of  random variables $(U,X)$. If $R_y=I(X;Y|U)$, then by~\eqref{eq:YS1_3} and the Markov structure~\eqref{eq:YS_joint}, we also have:
\begin{equation}
R_z\leq I(U;Y) \label{eq:YS3}
\end{equation}
and~\eqref{eq:YS1_1}, \eqref{eq:YS3}, \eqref{eq:YS1_2} coincide with the region ${\cal B}$. Therefore, we have only to examine the case where a strict inequality holds in~\eqref{eq:YS1_2}. Thus, let
\begin{equation}
R_y = I(X;Y|U) - \gamma \label{eq:YS4}
\end{equation}
for some $\gamma>0$. Define the random variable
\begin{equation}
U^* = \left\{ \begin{array}{ll}
                U & \mbox{w.p. } \lambda\\
                         X & \mbox{w.p. } 1-\lambda.
                         \end{array}\right.
                         \label{eq:YS4}
\end{equation}
Clearly, the Markov structure
\begin{equation}
 U^*- X- Y- Z \label{eq:YS4a}
 \end{equation}
still holds. Moreover
\begin{IEEEeqnarray}{rCl}
I(X;Y|U^*) &=& I(X;Y|U^*=U)P(U^*=U) + I(X;Y|U^*=X)P(U^*=X)\nonumber \\
                  &=& \lambda I(X;Y|U) \label{eq:YS5}
\end{IEEEeqnarray}
(In~\eqref{eq:YS5} and in the sequel, by $I(X;Y|U^*=U)$ we mean $I(X;Y|U^*,U^*=U)$, that is, the conditioning is not only on the event that $U^*=U$ but also on the specific value.) Now, we choose $\lambda$ to be
\begin{equation}
\lambda = \frac{I(X;Y|U) - \gamma}{I(X;Y|U)}.\label{eq:YS6}
\end{equation}
Note that with this choice, the following holds
\begin{equation}
R_y = I(X;Y|U) - \gamma = I(X;Y|U^*) \label{eq:YS7}
\end{equation}
and
\begin{IEEEeqnarray}{rCl}
R_y+R_z &\leq& I(X;Y) = I(XU^*;Y) = I(U^*;Y) + I(X;Y|U^*)\nonumber\\
 &=& I(U^*;Y) +R_y. \label{eq:YS8}
\end{IEEEeqnarray}
so that
\begin{equation}
R_z\leq I(U^*;Y). \label{eq:YS9}
\end{equation}
We now turn to bound $I(U^*;Z)$. For this purpose, observe that we can decompose $I(X;Z)$ as
\begin{subequations}
\label{eq:YS10}
\begin{IEEEeqnarray}{rCl}
I(X;Z) &=& I(U;Z) + I(X;Z|U)\label{eq:YS10_1}\\
          &=& I(U^*;Z) + I(X;Z|U^*)\nonumber\\
          &=& I(U^*;Z) + I(X;Z|U^*=U)P(U^*=U) \nonumber\\
          && +I(X;Z|U^*=X)P(U^*=X)\nonumber\\
          &=& I(U^*;Z) + \lambda I(X;Z|U).\label{eq:YS10_4}
\end{IEEEeqnarray}
\end{subequations}
From~\eqref{eq:YS10_1} and \eqref{eq:YS10_4} we obtain
\begin{equation}
I(U^*;Z) = I(U;Z) + (1-\lambda)I(X;Z|U) \geq I(U;Z). \label{eq:YS11}
\end{equation}
Therefore,~\eqref{eq:YS1_1} and~\eqref{eq:YS11} imply
\begin{equation}
R_z\leq I(U^*Z) + C_{12}.\label{eq:YS12}
\end{equation}
From~\eqref{eq:YS9}, \eqref{eq:YS12}, and~\eqref{eq:YS7} we have
\begin{subequations}
\label{eq:YS13}
\begin{IEEEeqnarray}{rCl}
R_z &\leq& \min \left\{ I(U^*;Z) + C_{12}, I(U^*;Y) \right\} \label{eq:YS13_1}\\
R_y &\leq& I(X;Y|U^*) \label{eq:YS13_2}
\end{IEEEeqnarray}
\end{subequations}
which, together with the Markov structure~\eqref{eq:YS4a}, imply that $(R_y,R_z)\in{\cal B}$.

\subsection{Proof of Achievability for Theorem \ref{Theorem1}}\label{Achievability1}
\par In this section, we prove the achievability part of Theorem \ref{Theorem1}. Throughout the achievability proof we use the definition of a strong typical set \cite{ElGammalKim10LectureNotes}. The set $\mathcal{T}_\epsilon^{(n)}(X,Y,Z)$ of $\epsilon$-typical $n$-sequences is defined by $\{(x^n,y^n,z^n):\frac{1}{n}|N(x,y,z|x^n,y^n,z^n)-p(x,y,z)|\leq \epsilon\cdot p(x,y,z)$ $\forall(x,y,z)\in \mathcal{X}\times\mathcal{Y}\times\mathcal{Z}\}$, where $N(x,y,z|x^n,y^n,z^n)$ is the number of appearances of $(x,y,z)$ in the $n$-sequence $(x^n,y^n,z^n)$.

\
\begin{proof}
\par Fix a joint distribution of $P_{S,U,X,Z,Y} = P_{S}P_{U|S}P_{X|S,U}P_{Y,Z|X,S}$ where $P_{Y,Z|X,S}=P_{Y|X,S}P_{Z|Y}$ is given by the channel. 

\par{\it Code Construction}: First, generate $2^{nC_{12}}$ superbins. Next, generate $2^{nR_Z}$ bins, one for each message $m_Z\in\{1,2,...,2^{nR_Z}\}$. Partition the bins among the superbins in their natural ordering such that each superbin $l\in\{1,2,...,2^{nC_{12}}\}$ contains the bins associated with the messages $m_Z \in\{(l-1)2^{n(R_Z-C_{12})}+1,...,l2^{n(R_Z-C_{12})}\}$. Thus, each superbin contains $2^{n(R_Z-C_{12})}$ bins. Second, for each bin generate $2^{n\tilde{R}_Z}$ codewords $u^n(m_Z,j)$, where $j\in\{1,2,...,2^{n\tilde{R}_Z}\}$. Each codeword, $u^n(m_Z,j)$, is generated i.i.d. $\sim\prod_{i=1}^n p(u_i)$. Third, for each codeword $u^n(m_Z,j)$ generate $2^{nR_Y}$ satellite bins. In each satellite bin generate $2^{n\tilde{R}_Y}$ codewords $x^n(m_Z,j,m_Y,k)$, where $k\in\{1,2,...,2^{n\tilde{R}_Y}\}$, i.i.d. $\sim\prod_{i=1}^n p(x_i(m_Z,j,m_Y,k)|u_i(m_Z,j))$. The code construction is illustrated in Fig. \ref{AchievabilityFig}.
\begin{figure}[h!]
\begin{center}
\begin{psfrags}
    \psfragscanon
    \psfrag{A}[][][0.9]{\ \ \ \ \ a codeword $u^n$}
    \psfrag{B}[][][0.7]{$1$}
    \psfrag{C}[][][0.7]{$2$}
    \psfrag{D}[][][0.55]{\ \ \ \ \ \ \ \ \ \ \ \ \ \ \ \ \ $2^{n(I(U;Z)-I(U;S)}$}
    \psfrag{E}[][][0.5]{}
    \psfrag{F}[][][0.9]{\ \ \ \ \ \ \ \ \ \ \ \ \ \ \ \ \ \ \ \ \ \ \ \ \ \ \ \ \ \ \ $2^{nI(U;S)}$ codewords in a bin}
    \psfrag{G}[][][0.9]{\ \ \ \ \ \ \ \ \ \ \ \ \ \ \ \ \ \ \ \ \ \ \ \ \ \ \ \ \ \ \ \ \ \ \ \ \ \ \ \ \ \ \ \ \ \ \ \ \ \ \ \ \ \ $2^{n\big{(}I(U;Z)-I(U;S)\big{)}}$ bins in a superbin}
    \psfrag{H}[][][0.9]{\ \ \ \ \ \ \ \ \ \ \ $2^{nC_{12}}$ superbins}
    \psfrag{I}[][][0.9]{\ \ \ \ \ \ \ \ \ \ \ \ \ \ \ \ \ \ \ \ \ \ \ \ \ \ \ \ \ \ \ $2^{nI(U;Z)}$ codewords in a superbin}
    \psfrag{J}[][][0.7]{\ \ \ \ \ \ $2^{nC_{12}}$}
    \psfrag{K}[][][0.9]{\ \ \ \ \ \ \ \ \ \ \ \ \ \ \ \ \ \ \ \ \ $u^n$ cloud center}
    \psfrag{L}[][][0.9]{\ \ \ \ \ \ \ \ \ \ \ \ \ \ \ \ \ \ \ \ \ \ \ \ \ \ \ \ \ \ \ \ \ \ \ \ \ \ \ \ \ \ \ \ \ \ \ \ \ \ \ \ \ \ \ \ \ \ \ \ \ \ \ \ \ \ \ \ \ \ \ \ \ \ \  $2^{nI(X;Y|U,S)}$ satellite bins. Each bin contains $2^{nI(X;S|U)}$ codewords $x^n$}
    \centerline{\includegraphics[scale = .8]{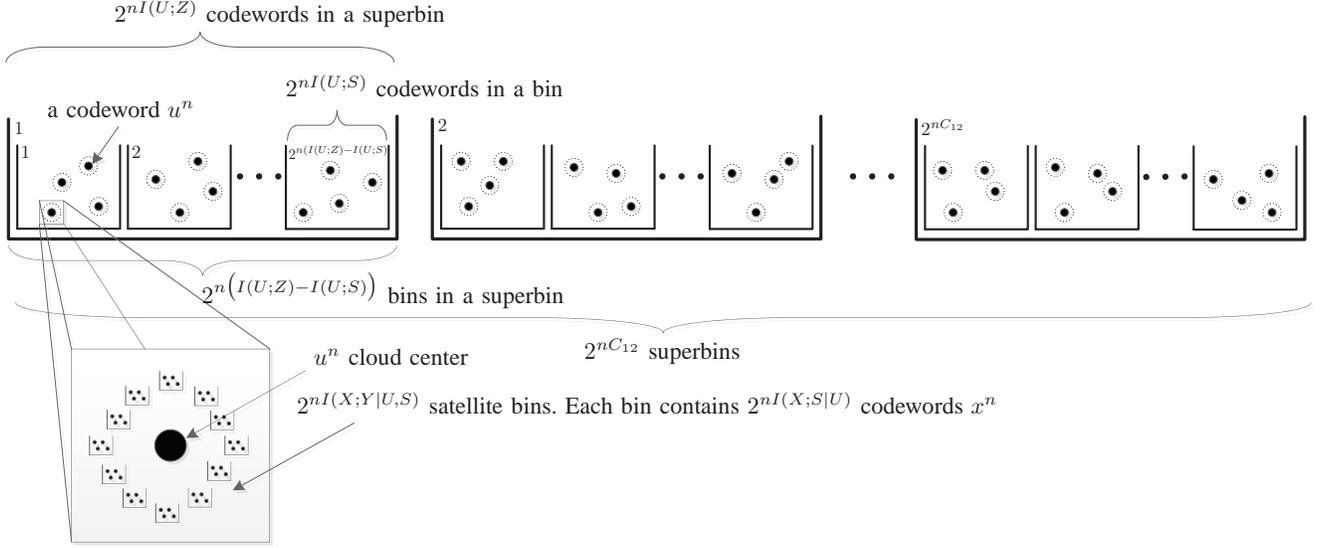}}
    \caption{The code construction. We have $2^{nC_{12}}$ superbins, each one containing $2^{n\big{(}I(U;Z)-I(U;S)\big{)}}$ bins. In each bin we have $2^{nI(U;S)}$ codewords $u^n$, so in total each superbin contains $2^{nI(U;Z)}$ codewords $u^n$. Finally, each codeword $u^n$ plays the role of a cloud center and is associated with $2^{nI(X;Y|U,S)}$  satellite codewords $x^n$. } \label{AchievabilityFig}
\end{psfrags}
\end{center}
\end{figure}

\par {\it Encoding}: To transmit $(m_Y,m_Z)$, the encoder first looks in the bin associated with the message $m_Z$ for a codeword $u^n(m_Z,j)$ such that it is jointly typical with the state sequence, $s^n$, i.e.
\begin{equation}
(u^n(m_Z,j),s^n)\in \mathcal{T}_{\epsilon'}^{(n)}(U,S).
\end{equation}
If such a codeword, $u^n$, does not exist, namely, no codeword in the bin $m_Z$ is jointly typical with $s^n$, choose an arbitrary $u^n$ from the bin (in such a case the decoder will declare an error). If there is more than one such codeword, choose the one for which $j$ is of the smallest lexicographical order. Next, the encoder looks for a sequence $x^n(m_Z,j,m_Y,k)$ (where $j$ was chosen in the first stage) such that it is jointly typical with the state sequence, $s^n$, and the codeword $u^n(m_Z,j)$, i.e.,
\begin{equation}
(x^n(m_Z,j,m_Y,k),u^n(m_Z,j),s^n)\in \mathcal{T}_{\epsilon'}^{(n)}(X,U,S).
\end{equation}
If such a codeword, $x^n$, does not exist, choose an arbitrary $x^n$ from the bin $m_Y$ (in such a case the decoder will declare an error). If there is more than one such codeword, choose the one for which $k$ is of the smallest lexicographical order.

\par {\it Decoding}:
\begin{enumerate}
  \item Let $\epsilon>\epsilon'$. Decoder Y looks for the smallest values of $(\hat{m}_Y,\hat{m}_Z)$ for which there exists a $\hat{j}$ and a $\hat{k}$ such that
  \begin{equation}
    (u^n(\hat{m}_Z,\hat{j})),x^n(\hat{m}_Z,\hat{j},\hat{m}_Y,\hat{k}),s^n,y^n)\in \mathcal{T}_{\epsilon}^{(n)}(U,X,S,Y).
  \end{equation}
    If no pair or more than one pair is found, an error is declared.
  \item Upon decoding the messages $(\hat{m}_Y,\hat{m}_Z)$, Decoder Y uses the link $C_{12}$ to send Decoder Z the superbin number, $\hat{l}$, that contains the decoded message $\hat{m}_Z$.
  \item Decoder Z looks in the superbin $\hat{l}$ for the smallest value of $\hat{m}_Z$  for which there exists a $\hat{j}$ such that
  \begin{equation}
    (u^n(\hat{m}_Z,\hat{j})),z^n)\in \mathcal{T}_{\epsilon}^{(n)}(U,Z).
\end{equation}
If no value or more than one value is found, an error is declared.
\end{enumerate}

\par{\bf Analysis of the probability of error:}\\
Without loss of generality, we can assume that messages $(m_Z,m_Y)=(1,1)$ were sent. Therefore, the superbin containing $m_Z=1$ is $l=1$.\\
We define the error events at the encoder:
\begin{align}
E_1 &= \{\forall j\in\{1,2,...,2^{n\tilde{R}_Z}\}: (U^n(1,j),S^n)\notin \mathcal{T}_{\epsilon'}^{(n)}(U,S)\},\\
E_2 &= \{\forall k\in\{1,2,...,2^{n\tilde{R}_Y}\}]: (X^n(1,1,j,k),U^n(1,j),S^n)\nonumber\\
& \ \ \ \ \ \notin \mathcal{T}_{\epsilon'}^{(n)}(X,U,S)\}.
\end{align}
We define the error events at Decoder Y:
\begin{align}
E_3 &= \{\forall j\in\{1,2,...,2^{n\tilde{R}_Z}\},\forall k\in\{1,2,...,2^{n\tilde{R}_Y}\}:\nonumber\\ &(U^n(j,1),X^n(1,j,1,k),S^n,Y^n)\notin \mathcal{T}_{\epsilon}^{(n)}(U,X,S,Y)\},\\
E_4 &= \{\exists \hat{m}_Y\neq 1: (U^n(j,1),X^n(1,j,\hat{m}_Y,k),S^n,Y^n)\nonumber\\
& \ \ \ \ \ \in \mathcal{T}_{\epsilon}^{(n)}(U,X,S,Y)\},\\
E_5 &= \{\exists \hat{m}_Z\neq 1, \hat{m}_Y\neq 1:\nonumber\\ &(U^n(j,\hat{m}_Z),X^n(\hat{m}_Z,j,\hat{m}_Y,k),S^n,Y^n)\in \mathcal{T}_{\epsilon}^{(n)}(U,X,S,Y)\},\\
E_6 &=\{\exists \hat{m}_Z\neq 1: \nonumber\\
&(U^n(j,\hat{m}_Z),X^n(\hat{m}_Z,j,1,k),S^n,Y^n)\in \mathcal{T}_{\epsilon}^{(n)}(U,X,S,Y)\}.
\end{align}

We define the error events at Decoder Z:
\begin{align}
E_7 &= \{\forall j\in\{1,2,...,2^{n\tilde{R}_Z}\}: (U^n(j,1),Z^n)\notin \mathcal{T}_{\epsilon}^{(n)}(U,Z)\},\\
E_8 &= \{\exists \hat{m}_Z\neq 1: m_Z \in\{1,2,...,2^{n(R_Z-C_{12})}\}\nonumber\\
 & \ \ \ \ \ (U^n(j,\hat{m}_Z),Z^n)\in \mathcal{T}_{\epsilon}^{(n)}(U,Z)\}.
\end{align}
Then, by the union bound:
\begin{align*}
P_e^{(n)} &\leq P(E_1)+P(E_2\cap E_1^c)+P(E_3\cap (E_1^c\cup E_2^c))+P(E_4)\nonumber\\
&+P(E_5)+P(E_6)+P(E_7\cap (E_1^c\cup E_2^c))+P(E_8).
\end{align*}
Now, consider:
\begin{enumerate}
  \item For the encoder error, using the Covering Lemma \cite{ElGammalKim10LectureNotes}, $P(E_1)$ tends to zero as $n\rightarrow \infty$ if in each bin associated with $m_Z$ we have more than $I(U;S)+\delta(\epsilon)$ codewords, i.e., $\tilde{R}_Z>I(U;S)+\delta(\epsilon)$.
  \item For the second term, we have that $(U^n,S^n)\in\mathcal{T}_{\epsilon}^{(n)}(U,S)$ and $X^n$ is generated i.i.d. $\sim\prod_{i=1}^n p(x_i|u_i)$. Hence, using the Covering Lemma, we have that $P(E_2\cap E_1^c)$ tends to zero as $n\rightarrow \infty$ if in each bin associated with $m_Y$ we have more than $I(X;S|U)+\delta(\epsilon)$ codewords, i.e., $\tilde{R}_Y>I(X;S|U)+\delta(\epsilon)$.
  \item For the third term, note that $(X^n, U^n, S^n)\in\mathcal{T}_{\epsilon'}^{(n)}(U,S,X)$. Furthermore, $Y^n$ is generated i.i.d. $\sim\prod_{i=1}^n p(y_i|x_i,s_i)$ and $\epsilon>\epsilon'$. Therefore, by the Conditional Typicality Lemma \cite{ElGammalKim10LectureNotes}, $P(E_3\cap (E_1^c\cup E_2^c))$ tends to zero as $n\rightarrow \infty$.
  \item For the fourth term, note that if $\hat{m}_Y\neq 1$ then for any $j\in\{1,2,...,2^{n\tilde{R}_Z}\}$ and any $k\in\{1,2,...,2^{n\tilde{R}_Y}\}$, $X^n(1,j,\hat{m}_Y,k)$ is conditionally independent of $(X^n(1,j,1,k),S^n,Y^n)$ given $U^n(1,j)$ and is distributed according to $\sim\prod_{i=1}^n p(x_i|u_i(1,j))$. Hence, by the Packing Lemma \cite{ElGammalKim10LectureNotes}, $P(E_4)$ tends to zero as $n\rightarrow \infty$ if $R_Y+\tilde{R}_Y< I(X;S,Y|U)-\delta(\epsilon)$.
  \item For the fifth term, note that for any $\hat{m}_Z\neq 1$, any $\hat{m}_Y\neq 1$, any $j\in\{1,2,...,2^{n\tilde{R}_Z}\}$ and any $k\in\{1,2,...,2^{n\tilde{R}_Y}\}$, $(U^n(\hat{m}_Z,j),X^n(\hat{m}_Z,j,\hat{m}_Y,k))$ are conditionally independent of $(U^n(1,j),X^n(1,j,1,k),S^n,Y^n)$. Hence, by the Packing Lemma \cite{ElGammalKim10LectureNotes}, $P(E_5)$ tends to zero as $m\rightarrow \infty$ if $R_Z+\tilde{R}_Z+R_Y+\tilde{R}_Y<I(U,X;Y,S)-\delta(\epsilon)$. This bound, in addition the the bounds on $\tilde{R}_Z$ and $\tilde{R}_Y$, gives us
      \begin{align*}
      R_Z+R_Y&<I(U,X;Y,S)-\tilde{R}_Z-\tilde{R}_Y-\delta(\epsilon)\\
             &<I(U,X;Y,S)-I(U;S)-I(X;S|U)-3\delta(\epsilon)\\
             &=I(U,X;Y|S)-3\delta(\epsilon)\\
             &=I(X;Y|S)-3\delta(\epsilon).
      \end{align*}
  \item For the sixth term, by the same considerations as for the previous event, by the Packing Lemma we have $R_Z-\tilde{R}_Z<I(U,X;Y,S)-\delta(\epsilon)$ (which is already satisfied).
  \item For the seventh term, $(X^n, U^n, S^n)\in\mathcal{T}_{\epsilon}^{(n)}(U,S,X)$. In addition, $Y^n$ is generated i.i.d. $\sim\prod_{i=1}^n p(y_i|x_i,s_i)$,  $Z^n$ is generated $\sim\prod_{i=1}^n p(z_i|y_i)=\prod_{i=1}^n p(z_i|y_i,x_i,s_i,u_i)$ and $\epsilon>\epsilon'$. Hence, by the Conditional Typicality Lemma \cite{ElGammalKim10LectureNotes} $P(E_7\cap (E_1^c\cup E_2^c))$ tends to zero as $n\rightarrow \infty$.
  \item For the eighth term,  note that for any $\hat{m}_Z\neq 1$ and any $j\in\{1,2,...,2^{n\tilde{R}_Z}\}$, $U^n(\hat{m}_Z,j)$ is conditionally independent of $(U^n(1,j),Z^n)$. Hence, by the Packing Lemma \cite{ElGammalKim10LectureNotes}, $P(E_8)$ tend to zero as $m\rightarrow \infty$ if the number of codewords in each superbin is less than $I(U;Z)$, i.e., $R_Z-C_{12}+\tilde{R}_Z<I(U;Z)-\delta(\epsilon)$.
\end{enumerate}

\par Combining the results, we have shown that $P(E)\rightarrow 0$ as $n\rightarrow \infty$ if
\begin{eqnarray}
R_Z&\leq& I(U;Z)-I(U;S)+C_{12}\nonumber\\
R_Y&\leq& I(X;Y|U,S)\nonumber\\
R_Z+R_Y&\leq& I(X;Y|S)\nonumber.
\end{eqnarray}
\par The above bound shows that the average probability of error, which, by symmetry, is equal to the probability for an individual pair of codewords, $(m_Z,m_Y)$, averaged over all choices of code-books in the random code construction, is arbitrarily small. Hence, there exists at least one code, $((2^{nR_Z},2^{nR_Y},2^{nR_{12}}),n)$, with an arbitrarily small probability of error.
\end{proof}

\subsection{Converse Proof of Theorem \ref{Theorem1}}\label{Converse}
In the previous section, we  proved the achievability part of  Theorem \ref{Theorem1}. In this section, we provide the upper bound on the capacity region of the PDSD BC, i.e., we give the proof of the converse for  Theorem \ref{Theorem1}.
\begin{proof}
Given an achievable rate trippet, $(R_Y,R_Z,C_{12})$, we need to show that there exists a joint distribution of the form (\ref{distribution1}), $P_{S}P_{U|S}P_{X|S,U}P_{Y|X,S}P_{Z|Y}$, such that
\begin{eqnarray}
R_Z&\leq& I(U;Z)-I(U;S)+C_{12}\nonumber\\
R_Y&\leq& I(X;Y|U,S)\nonumber\\
R_Z+R_Y&\leq& I(X;Y|S)\nonumber.
\end{eqnarray}
Since $(R_Y,R_Z,C_{12})$ is an achievable rate triplet, there exists a code, $(n, 2^{nR_Z}, 2^{nR_Y}, 2^{nC_{12}})$, with a probability
of error, $P^{(n)}_e$, that is arbitrarily small. By Fano's inequality,
\begin{equation}
H(M_Y|Y^n,S^n)\leq n(R_Y)P^{(n)}_{e,1}+H(P^{(n)}_{e,1})\triangleq \epsilon_{n_1},
\end{equation}
\begin{equation}
H(M_Z|Z^n,M_{12})\leq n(R_Z)P^{(n)}_{e,2}+H(P^{(n)}_{e,2})\triangleq \epsilon_{n_2},
\end{equation}
and let
\begin{equation}
\epsilon_{n_1}+\epsilon_{n_2}\triangleq \epsilon_n.
\end{equation}
Furthermore,
 \begin{eqnarray}
H(M_Y|M_Z,Y^n,S^n,Z^n)\leq H(M_Y|Y^n,S^n)\leq \epsilon_{n_1},\\
H(M_Z|Y^n,Z^n,S^n)\leq H(M_Z|Z^n,M_{12}(Y^n,S^n))\leq\epsilon_{n_2}.
\end{eqnarray}
Thus, can say that $\epsilon_n \rightarrow 0$ as $P^{(n)}_e \rightarrow 0$.

To bound the rate $R_Z$ consider:

\begin{align}
	nR_Z &= H(M_Z)\nonumber\\
	& =H(M_Z)-H(M_Z|Z^n,M_{12})+H(M_Z|Z^n,M_{12})\nonumber\\
	& \stackrel{(a)}{\leq}I(M_Z;Z^n,M_{12})+n\epsilon_n\nonumber\\
    & {=}I(M_Z;Z^n)+I(M_Z;M_{12}|Z^n)+n\epsilon_n\nonumber\\
    & \stackrel{(b)}{\leq}I(M_Z;Z^n)+H(M_{12})+n\epsilon_n\nonumber\\
    & \stackrel{(c)}{\leq}I(M_Z;Z^n)+C_{12}+n\epsilon_n\nonumber\\
    & {=} \sum_{i=1}^n I(M_Z;Z_i|Z^{i-1})+C_{12}+n\epsilon_n\nonumber\\
    & {\leq} \sum_{i=1}^n I(M_Z,Z^{i-1};Z_i)+C_{12}+n\epsilon_n\nonumber\\
    & {\leq} \sum_{i=1}^n I(M_Z,Z^{i-1},Y^{i-1};Z_i)+C_{12}+n\epsilon_n\nonumber\\
    & \stackrel{(d)}{=} \sum_{i=1}^n I(M_Z,Y^{i-1};Z_i)+C_{12}+n\epsilon_n\nonumber\\
    & {=} \sum_{i=1}^n I(M_Z,Y^{i-1},S_{i+1}^n;Z_i)-I(S_{i+1}^n;Z_i|M_Z,Y^{i-1})+C_{12}+n\epsilon_n\nonumber\\
    & \stackrel{(e)}{=} \sum_{i=1}^n I(M_Z,Y^{i-1},S_{i+1}^n;Z_i)-I(S_i;Y^{i-1}|M_Z,S_{i+1}^n)+C_{12}+n\epsilon_n\nonumber\\
    & \stackrel{(f)}{=} \sum_{i=1}^n I(M_Z,Y^{i-1},S_{i+1}^n;Z_i)-I(S_i;Y^{i-1},M_Z,S_{i+1}^n)+C_{12}+n\epsilon_n\nonumber\\
    & \stackrel{(g)}{=} \sum_{i=1}^n I(U_i;Z_i)-I(S_i;U_i)+C_{12}+n\epsilon_n\label{R2bound}
\end{align}
where\\
$(a)$ follows from Fano's inequality,\\
$(b)$ follows from the fact that conditioning reduces entropy,\\
$(c)$ follows from the admissibility of the conference,\\
$(d)$ follows from the physical degradedness properties of the channel,\\
$(e)$ follows from the Csisz아r sum identity,\\
$(f)$ follows from the fact that $S_i$ is independent of $(M_Z,S_{i+1}^n)$,\\
$(g)$ follows from the choice of $U_i = (M_Z,S_{i+1}^n,Y^{i-1})$.

\par Hence, we have:
\begin{equation}
R_Z\leq \frac{1}{n}\sum_{i=1}^n [I(U_i;Z_i)-I(S_i;U_i)]+C_{12}+\epsilon_n. \label{eq2}
\end{equation}
Next, to bound the rate $R_Y$ consider:
\begin{align}
	nR_Y &= H(M_Y)\nonumber\\
    &\stackrel{(a)}{=} H(M_Y|M_Z,S^n)\nonumber\\
	& =H(M_Y|M_Z,S^n)-H(M_Y|M_Z,S^n,Y^n)+H(M_Y|M_Z,S^n,Y^n)\nonumber\\
	& \stackrel{(b)}{\leq}I(M_Y;Y^n|M_Z,S^n)+n\epsilon_n\nonumber\\
    & \stackrel{(c)}{=}I(M_Y,X^n(M_Y,M_Z,S^n);Y^n|M_Z,S^n)+n\epsilon_n\nonumber\\
    & {=} \sum_{i=1}^n I(M_Y,X^n;Y_i|M_Z,S^n,Y^{i-1})+n\epsilon_n\nonumber\\
    & {=} \sum_{i=1}^n H(Y_i|M_Z,S^n,Y^{i-1})-H(Y_i|M_Z,S^n,Y^{i-1},M_Y,X^n)+n\epsilon_n\nonumber\\
    & \stackrel{(d)}{\leq} \sum_{i=1}^n H(Y_i|M_Z,S_i,S_{i+1}^n,Y^{i-1})-H(Y_i|M_Z,S^n,Y^{i-1},M_Y,X^n)+n\epsilon_n\nonumber\\
    & \stackrel{(e)}{\leq} \sum_{i=1}^n H(Y_i|M_Z,S_i,S_{i+1}^n,Y^{i-1})-H(Y_i|M_Z,S_i,S_{i+1}^n,Y^{i-1},X_i)+n\epsilon_n\nonumber\\
    & {=} \sum_{i=1}^n I(Y_i;X_i|M_Z,S_i,S_{i+1}^n,Y^{i-1})+n\epsilon_n\nonumber\\
    & \stackrel{(f)}{=} \sum_{i=1}^n I(Y_i;X_i|S_i,U_i)+n\epsilon_n\nonumber\\
\end{align}
where\\
$(a)$ follows from the fact that $M_Y$ is independent of $(M_Z,S^n)$,\\
$(b)$ follows from Fano's inequality,\\
$(c)$ follows from the fact that $X^n$ is a deterministic function of $(M_Z,M_Y,S^n)$,\\
$(d)$ follows from the fact that conditioning reduces entropy,\\
$(e)$ follows from the properties of the channel,\\
$(f)$ follows from the choice of $U_i = (M_Z,S_{i+1}^n,Y^{i-1})$.

\par Hence, we have:
\begin{equation}
R_Y\leq \frac{1}{n}\sum_{i=1}^n I(Y_i;X_i|S_i,U_i)+\epsilon_n. \label{eq2}
\end{equation}

\par To bound the sum of rates, $R_Z+R_Y$, consider:
\begin{align*}
	n(R_Z+R_Y) &= H(M_Z,M_Y)\\
    &\stackrel{(a)}{=} H(M_Z,M_Y|S^n)\\
    & =H(M_Z,M_Y|S^n)+H(M_Z,M_Y|Y^n,Z^n,S^n)-H(M_Z,M_Y|Y^n,Z^n,S^n) \\
    & \stackrel{(b)}{\leq}I(M_Z,M_Y;Y^n,Z^n|S^n)+n\epsilon_n\nonumber\\
    & \stackrel{(c)}{=}I(M_Z,M_Y;Y^n|S^n)+n\epsilon_n\nonumber\\
    & \stackrel{(d)}{=}I(M_Y,M_Z,X^n(M_Y,M_Z,S^n);Y^n|S^n)+n\epsilon_n\nonumber\\
    & {=} \sum_{i=1}^n I(M_Y,M_Z,X^n;Y_i|S^n,Y^{i-1})+n\epsilon_n\nonumber\\
    & {=} \sum_{i=1}^n H(Y_i|S^n,Y^{i-1})-H(Y_i|S^n,Y^{i-1},X^n,M_Y,M_Z)+n\epsilon_n\nonumber\\
    & \stackrel{(e)}{\leq} \sum_{i=1}^n H(Y_i)-H(Y_i|S^n,Y^{i-1},X^n,M_Y,M_Z)+n\epsilon_n\nonumber\\
    & \stackrel{(f)}{\leq} \sum_{i=1}^n H(Y_i)-H(Y_i|S_i,X_i)+n\epsilon_n\nonumber\\
    & {=} \sum_{i=1}^n I(Y_i;X_i|S_i)+n\epsilon_n\nonumber\\
\end{align*}

where\\
$(a)$ follows from the fact that $(M_Z,M_Y)$ are independent of $S^n$,\\
$(b)$ follows from Fano's inequality,\\
$(c)$ follows from the physical degradedness and memorylessness of the channel,\\
$(d)$ follows from the fact that $X^n$ is a deterministic function of $(M_Z,M_Y,S^n)$,\\
$(e)$ follows from the fact that conditioning reduces entropy,\\
$(f)$ follows from the properties of the channel.

\par Hence, we have:
\begin{equation}
R_Z+R_Y\leq \frac{1}{n}\sum_{i=1}^n I(Y_i;X_i|S_i)+\epsilon_n. \label{eq1}
\end{equation}

\par We complete the proof by using standard time-sharing arguments to obtain the rate bounds terms given in (\ref{CapacityRegion1}).
\end{proof}


\subsection{Proof of Theorem \ref{Theorem2}}\label{SectionCausalProof}                    
\begin{proof}
The proof of Theorem \ref{Theorem2} follows closely the proof of the noncausal case. Therefore, the proof is not given in full detail. Instead, we rely on the guidelines of the proof for the noncausal case and emphasize the differences when considering the causal scenario.
\par For the proof of achievability, fixing a distribution of the form (\ref{distribution2}), we generate codewords $u^n(m_Z)$ i.i.d. $\sim\prod_{i=1}^n p(u_i)$. Next, we generate satellite codewords $v^n(m_Y,m_Z)$ i.i.d. $\sim\prod_{i=1}^n p(v_i|u_i)$ (instead of $x^n(m_Y,m_Z)$) around each cloud center $u^n(m_Z)$. Furthermore, the codewords $u^n$ are divided among $2^{nC_{12}}$ superbins. To send $(m_Y,m_Z)$, the encoder transmits $x_i(u_i(m_z),v_i(m_Z,m_Y),s_i)$ at time $i\in[1,n]$. For decoding, the strong decoder, Decoder Y, decodes the satellite codeword $v^n$ (and hence also the cloud center $u^n$), i.e., it looks for an $(\hat{m}_Y,\hat{m}_Z)$ such that $(u^n(\hat{m}_Z),v^n(\hat{m}_Z,\hat{m}_Y),y^n)\in \mathcal{T}_{\epsilon}^{(n)}(U,V,Y)$. Decoder Y then uses the link between the decoders to send Decoder Z the number of the superbin that contains $u^n$. Decoder Z now looks in this specific superbin for a unique $\hat{m}_Z$ such that $(u^n(\hat{m}_Z),z^n)\in \mathcal{T}_{\epsilon}^{(n)}(U,Z)$. Now, by LLN, the Conditional Typicality Lemma, the Packing Lemma \cite{ElGammalKim10LectureNotes}, and similar to the proof of achievability of Theorem \ref{Theorem1}, the probability of error tends to zero as $n\rightarrow \infty$ if
\begin{eqnarray*}
R_Z&\leq& I(U;Z)+C_{12}-\delta(\epsilon)\\
R_Y&\leq& I(V;Y|U)-\delta(\epsilon)\\
R_Z+R_Y&\leq& I(V,U;Y)-\delta(\epsilon).
\end{eqnarray*}
\par For the converse, we define the auxiliary random variables $U_i=(M_Z,Y^{i-1})$ and $V_i=M_Y$. Note that for this setting, this definition of $U_i$ and $V_i$ result in $(U_i,V_i)$ that are independent of $S_i$. Therefore, if we follow the same steps as in the converse of Theorem \ref{Theorem1}, the bound on $R_Z$ reduces to
\begin{equation}
R_Z\leq \frac{1}{n}\sum_{i=1}^n I(U_i;Z_i)+C_{12}+\epsilon_n.
\end{equation}
Next, to bound the rate $R_Y$ consider:
\begin{align}
	nR_Y &= H(M_Y)\nonumber\\
	& {\leq}I(M_Y;Y^n|M_Z)+n\epsilon_n\nonumber\\
    & {=} \sum_{i=1}^n I(M_Y;Y_i|M_Z,Y^{i-1})+n\epsilon_n\nonumber\\
    & {=} \sum_{i=1}^n I(V_i;Y_i|U_i)+n\epsilon_n\nonumber.
\end{align}
\par Hence, we have:
\begin{equation}
R_Y\leq \frac{1}{n}\sum_{i=1}^n I(V_i;Y_i|U_i)+\epsilon_n. \label{eq2}
\end{equation}
Finally, to bound the sum of rates, $R_Z+R_Y$, consider:
\begin{align*}
	n(R_Z+R_Y) &= H(M_Z,M_Y)\\
    & {\leq}I(M_Z,M_Y;Y^n)+n\epsilon_n\nonumber\\
    & {=} \sum_{i=1}^n I(M_Y,M_Z;Y_i|Y^{i-1})+n\epsilon_n\nonumber\\
    & {\leq} \sum_{i=1}^n I(M_Y,M_Z,Y^{i-1};Y_i)+n\epsilon_n\nonumber\\
    & {=} \sum_{i=1}^n I(V_i,U_i;Y_i)+n\epsilon_n\nonumber.
\end{align*}
\par Hence, we have:
\begin{equation}
R_Z+R_Y\leq \frac{1}{n}\sum_{i=1}^n I(V_i,U_i;Y_i)+\epsilon_n. \label{eq2}
\end{equation}
We complete the proof by using standard time-sharing arguments; hence, the details are omitted.
\end{proof}

\subsection{Proof of Achievability of Theorem \ref{Theorem3}}\label{Achievability2}
\par Let us prove achievability of the region given in Theorem \ref{Theorem3}.

\begin{proof}
\par Fix a joint distribution of the form (\ref{distribution3}), $P_{S,S_d,U,X,Z,Y} = P_{S}P_{S_d,U,X|S}P_{Y,Z|X,S}$, where $P_{Y,Z|X,S}=P_{Y|X,S}P_{Z|Y}$ is given by the channel.

\par{\it Code Construction}:
First, we start by generating the codebook of the State Encoder. Randomly and independently generate $2^{n\tilde{R}_{12}}$ sequences $s_d^n(l)$, $l\in[1,2^{n\tilde{R}_{12}}]$ i.i.d. $\sim \prod_{i=1}^n p(s_{d,i})$. Partition the codewords, $s_d^n(l)$, among $2^{nC_{12}}$ bins in their natural ordering such that each bin $B(t)$, $t\in[1,2^{nC_{12}}]$ contains the codewords associated with the index $l\in[(t-1)2^{n(\tilde{R}_{12}-C_{12})}+1,t2^{n(\tilde{R}_{12}-C_{12})}]$. Reveal the codebook to the Channel Encoder, Decoder Y and Decoder Z.
\par Second, we create the codebook for the Channel encoder. Generate $2^{nR_Z}$ bins, $B(m_Z)$, $m_Z\in[1,2^{nR_Z}]$. In each bin generate $2^{n\tilde{R}_Z}$ codewords $u^n(j,m_Z)$, $j\in[1,2^{n\tilde{R}_Z}]$ i.i.d. $\sim \prod_{i=1}^n p(u_i)$. Third, for each codeword $u^n(j,m_Z)$  generate $2^{nR_Y}$ satellite bins. In each satellite bin generate $2^{n\tilde{R}_Y}$ codewords $x^n(m_Z,j,m_Y,k)$, where $k\in[1,2^{n\tilde{R}_Y}]$, i.i.d. $\sim\prod_{i=1}^n p(x_i(m_Z,j,m_Y,k)|u_i(m_Z,j))$.


\par {\it Encoding}:
\begin{enumerate}
  \item {\it State Encoder}: Given $s^n$, the State Encoder finds an index $l$ such that
    \begin{eqnarray}
    (s_d^n(l),s^n)\in \mathcal{T}_{\epsilon'}^{(n)}(S_d,S).
    \end{eqnarray}
    If there is more than one such index, choose the one for which $l$ is of the smallest lexicographical order. If there is no such index, select an index at random from the bin $B(t)$. The State Encoder sends the bin index $t$.
  \item {\it Channel Encoder}: First, note that the Channel Encoder knows the sequence transmitted from the State Encoder, $s_d^n(l)$, since it knows both $s^n$ and the State Encoder's strategy. To transmit $(m_Y,m_Z)$, the encoder first looks in the bin associated with the message $m_Z$ for a codeword $u^n(j,m_Z)$ such that it is jointly typical with the state sequence, $s^n$, and the codeword $s_d^n(l)$, i.e.
      \begin{eqnarray}
       (u^n(j,m_Z),s^n,s_d^n(l))\in \mathcal{T}_{\epsilon'}^{(n)}(U,S,S_d).
      \end{eqnarray}
      If there is more than one such index, choose the one for which $j$ is of the smallest lexicographical order. If there is no such index, choose an arbitrary $u^n$ from the bin $m_Z$ (in such a case the decoder will declare an error). Next, the encoder looks for a sequence $x^n(m_Z,j,m_Y,k)$ (where $j$ was chosen in the first stage) such that it is jointly typical with the state sequence, $s^n$, the codeword $s_d^n(l)$ and the codeword $u^n(m_Z,j)$, i.e.,
        \begin{equation}
        (x^n(m_Z,j,m_Y,k),u^n(m_Z,j),s_d^n(l),s^n)\in \mathcal{T}_{\epsilon'}^{(n)}(X,U,S,S_d).
        \end{equation}
        If such a codeword, $x^n$, does not exist, choose an arbitrary $x^n$ from the bin $m_Y$ (in such a case the decoder will declare an error). If there is more than one such codeword, choose the one for which $k$ is of the smallest lexicographical order.

\end{enumerate}

\par {\it Decoding}:
\begin{enumerate}
  \item Let $\epsilon>\epsilon'$. Note that Decoder Y knows both the sequence $s_d^n$ and $s^n$. Since Decoder Y knows the sequence $s^n$ and the State Encoder's strategy, it also knows $s_d^n$ (similar to the Channel Encoder). Hence, it looks for the smallest values of $(\hat{m}_Y,\hat{m}_Z)$ for which there exists a $\hat{j}$ such that
  \begin{equation}
    (u^n(\hat{m}_Z,\hat{j})),x^n(\hat{m}_Z,\hat{j},\hat{m}_Y,\hat{k}),s_d^n,s^n,y^n)\in \mathcal{T}_{\epsilon}^{(n)}(U,X,S_d,S,Y).
  \end{equation}
    If no triplet or more than one such triplet is found, an error is declared.
  \item Decoder Z first looks for the unique index $\hat{l}\in B(t)$ such that
  \begin{equation}
    (s_d^n(\hat{l}),z^n)\in \mathcal{T}_{\epsilon}^{(n)}(S_d,Z).
  \end{equation}
  \item Once Decoder Z has decoded $s_d^n(\hat{l})$, it uses $s_d^n(\hat{l})$ as side information to help the next decoding stage. Hence, the second step is to look for the smallest value of $\hat{m}_Z$  for which there exists a $\hat{j}$ such that
  \begin{equation}
    (u^n(\hat{m}_Z,\hat{j})),z^n,s_d(\hat{l}))\in \mathcal{T}_{\epsilon}^{(n)}(U,Z,S_d).
\end{equation}
If no pair or more than one such pair is found, an error is declared.
\end{enumerate}

\par{\bf Analysis of the probability of error:}\\
Without loss of generality, we can assume that messages $(m_Z,m_Y)=(1,1)$ were sent. We define the error event at the State Encoder:
\begin{eqnarray}
E_1 &=& \{\forall l\in[1,2^{n\tilde{R}_{12}}]: (S_d^n(l),S^n)\notin \mathcal{T}_{\epsilon'}^{(n)}(S_d,S)\}.
\end{eqnarray}
We define the error events at the Channel Encoder:
\begin{eqnarray}
E_2 &=& \{\forall j\in[1,2^{n\tilde{R}_Z}]: (U^n(1,j),S^n,S_d^n)\notin \mathcal{T}_{\epsilon'}^{(n)}(U,S,S_d)\},\\
E_3 &=& \{\forall k\in[1,2^{n\tilde{R}_Y}]: (X^n(1,j,1,k),U^n(1,j),S^n,S_d^n)\notin \mathcal{T}_{\epsilon'}^{(n)}(X,U,S,S_d)\}.
\end{eqnarray}
We define the error events at Decoder Y:
\begin{eqnarray}
E_4 &=& \{\forall j\in[1,2^{n\tilde{R}_Z}]: (U^n(j,1),X^n(1,j,1,k),S_d^n,S^n,Y^n)\notin \mathcal{T}_{\epsilon}^{(n)}(U,X,S_d,S,Y)\},\\
E_5 &=& \{\exists \hat{m}_Y\neq 1: (U^n(j,1),X^n(1,j,\hat{m}_Y,k),S_d^n,S^n,Y^n)\in \mathcal{T}_{\epsilon}^{(n)}(U,X,S_d,S,Y)\},\\
E_6 &=& \{\exists \hat{m}_Z\neq 1, \hat{m}_Y\neq 1: (U^n(j,\hat{m}_Z),X^n(\hat{m}_Z,j,\hat{m}_Y,k),S_d^n,S^n,Y^n)\in \mathcal{T}_{\epsilon}^{(n)}(U,X,S_d,S,Y)\},\\
E_7 &=&\{\exists \hat{m}_Z\neq 1: (U^n(j,\hat{m}_Z),X^n(\hat{m}_Z,j,1,k),S_d^n,S^n,Y^n)\in \mathcal{T}_{\epsilon}^{(n)}(U,X,S_d,S,Y)\}.
\end{eqnarray}
We define the error events at Decoder Z:
\begin{eqnarray}
E_8 &=& \{\forall l\in[1,2^{n\tilde{R}_{12}}]: (S_d^n(l),Z^n)\notin \mathcal{T}_{\epsilon}^{(n)}(S_d,Z)\},\\
E_9 &=& \{\exists \hat{l}\neq L, \hat{l}\in B(T):(S_d^n(\hat{l}),Z^n)\in \mathcal{T}_{\epsilon}^{(n)}(S_d,Z)\},\\
E_{10} &=& \{\forall j\in[1,2^{n\tilde{R}_Z}]: (U^n(j,1),Z^n,S_d^n(l))\notin \mathcal{T}_{\epsilon}^{(n)}(U,Z,S_d)\},\\
E_{11} &=& \{\exists \hat{m}_Z\neq 1, (U^n(j,\hat{m}_Z),Z^n,S_d^n(l))\in \mathcal{T}_{\epsilon}^{(n)}(U,Z,S_d)\}.
\end{eqnarray}
Then, by the union of events bound:
\begin{align*}
P_e^{(n)} \leq &P(E_1)+P(E_2\cap E_1^c)+P(E_3\cap (E_1^c\cup E_2^c))+P(E_4\cap (E_1^c\cup E_2^c\cup E_3^c))+P(E_5)+P(E_6)+P(E_7)\\
&+P(E_8\cap (E_1^c\cup E_2^c\cup E_3^c))+P(E_9)+P(E_{10}\cap (E_1^c\cup E_2^c\cup E_3^c))+P(E_{11}\cap E_1^c).
\end{align*}

Now, consider:
\begin{enumerate}
  \item For the error at the State Encoder, $E_1$, by invoking the Covering Lemma \cite{ElGammalKim10LectureNotes}, we obtain that $P(E_1)$ tends to zero as $n\rightarrow \infty$ if $\tilde{R}_{12}\geq I(S;S_d)+\delta(\epsilon')$.
  \item The probabilities of the errors $P(E_2\cap E_1^c),P(E_3\cap (E_1^c\cup E_2^c)),P(E_4\cap (E_1^c\cup E_2^c\cup E_3^c)),P(E_5),P(E_6)$ and $P(E_7)$ are treated in the exact same manner as in Section \ref{Achievability1}, where they are shown to tend to zero as $n\rightarrow \infty$ . Therefore the details are omitted.
  \item For the eighth term, we have that $(X^n, U^n, S^n, S_d^n)\in\mathcal{T}_{\epsilon}^{(n)}(X,U,S,S_d)$. In addition, $Y^n$ is generated i.i.d. $\sim\prod_{i=1}^n p(y_i|x_i,s_i)$, and $Z^n$ is generated $\sim\prod_{i=1}^n p(z_i|y_i)=\sim\prod_{i=1}^n p(z_i|y_i,x_i,s_i,u_i,s_{d,i})$. Since $\epsilon>\epsilon'$, by the Conditional Typicality Lemma \cite{ElGammalKim10LectureNotes}, $P(E_8\cap (E_1^c\cup E_2^c\cup E_3^c))$ tends to zero as $n\rightarrow\infty$.
  \item For the ninth error expression, $E_9$, we have that
        \begin{eqnarray}
        P(E_9)&=&P(\exists \hat{l}\neq L, \hat{l}\in B(T):(S_d^n(\hat{l}),Z^n)\in \mathcal{T}_{\epsilon}^{(n)}(S_d,Z))\nonumber\\
        &\leq& P(\exists  \hat{l}\in B(1):(S_d^n(\hat{l}),Z^n)\in \mathcal{T}_{\epsilon}^{(n)}(S_d,Z))\nonumber
        \end{eqnarray}
        \cite[{Lemma 11.1.}]{ElGammalKim10LectureNotes}. Therefore, since the sequence $S_d^n(\hat{l})$ is independent of $Z^n$, by the Packing Lemma \cite{ElGammalKim10LectureNotes}, $P(E_9)$ tends to zero as $n\rightarrow \infty$ if $\tilde{R}_{12}-C_{12}\leq I(Z;S_d)+\delta(\epsilon)$.
  \item For the tenth term, we again note that the random variables $(U,Z,S,X,S_d)$ are generated i.i.d.. Hence, since $\epsilon>\epsilon'$ and by the Conditional Typicality Lemma \cite{ElGammalKim10LectureNotes}, $P(E_{10}\cap (E_1^c\cup E_2^c\cup E_3^c))$ tends to zero as $m\rightarrow \infty$.
  \item For the eleventh term,  note that for any $\hat{m}_Z\neq 1$ and any $j\in[1,2^{n\tilde{R}_Z}]$, $U^n(\hat{m}_Z,j)$ is conditionally independent of $(U^n(1,j),Z^n,S_d^n)$. Hence, by the Packing Lemma \cite{ElGammalKim10LectureNotes}, $P(E_{11}\cap E_1^c)$ tend to zero as $m\rightarrow \infty$ if $R_Z-\tilde{R}_Z<I(U;Z,S_d)-\delta(\epsilon)$
\end{enumerate}

\par Combining the results, we have shown that $P(E)\rightarrow 0$ as $n\rightarrow \infty$ if
\begin{eqnarray}
R_Z&\leq& I(U;Z,S_d)-I(U;S,S_d)\nonumber\\
R_Y&\leq& I(X;S,S_d,Y|U)-I(X;S,S_d|U)\nonumber\\
R_Z+R_Y&\leq& I(U,X;S,S_d,Y)-I(U,X;S,S_d)\nonumber\\
C_{12}&\geq& I(S;S_d)-I(Z;S_d).\nonumber
\end{eqnarray}

\begin{remark}
Rearranging the expressions, we obtain
\begin{eqnarray}
R_Z&\leq& I(U;Z,S_d)-I(U;S,S_d)\label{RzBound1}\\
R_Y&\leq& I(X;Y|U,S,S_d)\label{RyBound1}\\
R_Z+R_Y&\leq& I(U,X;Y|S,S_d)\label{redundantBound1}\\
C_{12}&\geq& I(S;S_d)-I(Z;S_d).
\end{eqnarray}
Note that the bound on the rate sum, (\ref{redundantBound1}), is redundant and can be removed, since:
\begin{align*}
R_Z+R_Y&\leq I(U,X;Y|S,S_d)\\
       &= I(U;Y|S,S_d)+I(X;Y|U,S,S_d)\\
       &= I(U;Y,S,S_d)-I(U;S_d,S)+I(X;Y|U,S,S_d),
\end{align*}
in addition to $R_Y$ satisfying (\ref{RyBound1}) and $R_Z$ satisfying (\ref{RzBound1}), where (\ref{RzBound1}) can be bounded by
\begin{align*}
R_Z&\leq I(U;Z,S_d)-I(U;S,S_d)\\
   &\leq I(U;Y,S_d)-I(U;S,S_d)\\
   &\leq I(U;Y,S_d,S)-I(U;S,S_d).
\end{align*}
\end{remark}
\par The above bound shows that the average probability of error, which, by symmetry, is equal to the probability for an individual pair of codewords, $(m_Z,m_Y)$, averaged over all choices of code-books in the random code construction, is arbitrarily small. Hence, there exists at least one code, $((2^{nR_Z},2^{nR_Y},2^{nR_{12}}),n)$, with an arbitrarily small probability of error.
\end{proof}

\subsection{Converse Proof of Theorem \ref{Theorem3}}\label{Converse2}
In Section \ref{Achievability2}, the achievability Theorem \ref{Theorem3} was shown. To finish the proof, we provide the upper bound on the capacity region.
\begin{proof}
Given an achievable rate triplet $(C_{12},R_Z,R_Y)$, we need to show that there exists a joint distribution of the form (\ref{distribution3}), $P_{S}P_{S_d,U,X|S}P_{Y|X,S}P_{Z|Y}$, such that
\begin{eqnarray}
R_Z&\leq& I(U;Z|S_d)-I(U;S|S_d)\nonumber\\
R_Y&\leq& I(X;Y|U,S)\nonumber
\end{eqnarray}
and
\begin{eqnarray}
C_{12}&\geq& I(S;S_d)-I(Z;S_d).\nonumber
\end{eqnarray}
Since $(C_{12},R_Z,R_Y)$ is an achievable rate triplet, there exists a code, $((2^{nR_Z},2^{nR_Y}),2^{nC_{12}},n)$, with a probability
of error, $P^{(n)}_e$, that is arbitrarily small. By Fano's inequality,
\begin{equation}
H(M_Y|Y^n,S^n)\leq n(R_Y)P^{(n)}_{e,1}+H(P^{(n)}_{e,1})\triangleq \epsilon_{n_1},
\end{equation}
\begin{equation}
H(M_Z|Z^n,M_{12})\leq n(R_Z)P^{(n)}_{e,2}+H(P^{(n)}_{e,2})\triangleq \epsilon_{n_2},
\end{equation}
and let
\begin{equation}
\epsilon_{n_1}+\epsilon_{n_2}\triangleq \epsilon_n.
\end{equation}
Furthermore,
 \begin{eqnarray}
H(M_Y|M_Z,Y^n,S^n,Z^n)\leq H(M_Y|Y^n,S^n)\leq \epsilon_{n_1},\\
H(M_Z|Y^n,Z^n,S^n,M_{12})\leq H(M_Z|Z^n,M_{12}))\leq\epsilon_{n_2}.
\end{eqnarray}
Thus, we can say that $\epsilon_n \rightarrow 0$ as $P^{(n)}_e \rightarrow 0$.

For $C_{12}$ consider:

\begin{align}
	nC_{12} &\geq H(M_{12})\nonumber\\
	& {\geq}I(M_{12};S^n)\nonumber\\
    & {=} \sum_{i=1}^n I(M_{12};S_i|S^{i-1})\nonumber\\
    & \stackrel{(a)}{=}\sum_{i=1}^n I(M_{12},S^{i-1};S_i)\nonumber\\
    & {=} \sum_{i=1}^n I(M_{12},S^{i-1},Z_{i+1}^n;S_i)-I(Z_{i+1}^n;S_i|M_{12},S^{i-1})\nonumber\\
    & \stackrel{(b)}{=} \sum_{i=1}^n I(M_{12},S^{i-1},Z_{i+1}^n;S_i)-I(Z_i;S^{i-1}|M_{12},Z_{i+1}^n)\nonumber\\
    & {\geq} \sum_{i=1}^n I(M_{12},S^{i-1},Z_{i+1}^n;S_i)-I(Z_i;S^{i-1},M_{12},Z_{i+1}^n)\nonumber\\
    & \stackrel{(c)}{=} \sum_{i=1}^n I(S_{d,i};S_i)-I(Z_i;S_{d,i})\nonumber\label{R2bound}
\end{align}
where\\
$(a)$ follows since $S_i$ is independent of $S^{i-1}$,\\
$(b)$ follows from the Csisz아r sum identity,\\
$(c)$ follows from the definition of the auxiliary random variable, $S_d= (M_{12},S_{i+1}^n,Z^{i-1})$.

\par Hence, we have:
\begin{equation}
C_{12}\leq \frac{1}{n}\sum_{i=1}^n I(S_{d,i};S_i)-I(Z_i;S_{d,i}). \label{eq4}
\end{equation}

To bound the rate $R_Z$ consider:

\begin{align}
	nR_Z &= H(M_Z)\nonumber\\
    &= H(M_Z|M_{12})\nonumber\\
	& =H(M_Z|M_{12})-H(M_Z|Z^n,M_{12})+H(M_Z|Z^n,M_{12})\nonumber\\
	& \stackrel{(a)}{\leq}I(M_Z;Z^n|M_{12})+n\epsilon_n\nonumber\\
    & {=} \sum_{i=1}^n I(M_Z;Z_i|M_{12},Z^{i-1})+n\epsilon_n\nonumber\\
    & {=} \sum_{i=1}^n I(M_Z,S_{i+1}^n;Z_i|M_{12},Z^{i-1})-I(S_{i+1}^n;Z_i|M_Z,M_{12},Z^{i-1})+n\epsilon_n\nonumber\\
    & \stackrel{(b)}{=} \sum_{i=1}^n I(M_Z,S_{i+1}^n;Z_i|M_{12},Z^{i-1})-I(S_i;Z^{i-1}|M_Z,M_{12},S_{i+1}^n)+n\epsilon_n\nonumber\\
    & \stackrel{(c)}{=} \sum_{i=1}^n I(M_Z,S_{i+1}^n;Z_i|M_{12},Z^{i-1})-I(S_i;Z^{i-1},M_Z|M_{12},S_{i+1}^n)+n\epsilon_n\nonumber\\
    & \stackrel{(d)}{=} \sum_{i=1}^n I(S_{i+1}^n;Z_i|M_{12},Z^{i-1})+ I(M_Z,;Z_i|M_{12},Z^{i-1},S_{i+1}^n)-I(S_i;Z^{i-1},M_Z|M_{12},S_{i+1}^n)+n\epsilon_n\nonumber\\
    & \stackrel{(e)}{=} \sum_{i=1}^n I(S_i;Z^{i-1}|M_{12},S_{i+1}^n)+ I(M_Z,;Z_i|M_{12},Z^{i-1},S_{i+1}^n)-I(S_i;Z^{i-1},M_Z|M_{12},S_{i+1}^n)+n\epsilon_n\nonumber\\
    & {=} \sum_{i=1}^n I(M_Z,;Z_i|M_{12},Z^{i-1},S_{i+1}^n)-I(S_i;M_Z|Z^{i-1},M_{12},S_{i+1}^n)+n\epsilon_n\nonumber\\
    & \stackrel{(f)}{=} \sum_{i=1}^n I(U_i;Z_i|S_{d,i})-I(S_i;U_i|S_{d,i})+n\epsilon_n\label{R2bound}
\end{align}
where\\
$(a)$ follows from Fano's inequality,\\
$(b)$ follows from the Csisz아r sum identity,\\
$(c)$ follows from the fact that $S_i$ is independent of $M_Z$ given $(M_Z,M_{12},S_{i+1}^n)$,\\
$(d)$ follows from using the chain rule,\\
$(e)$ follows from using the Csisz아r sum identity,\\
$(f)$ follows from the definition of the auxiliary random variables $S_d= (M_{12},S_{i+1}^n,Z^{i-1})$ and $U_i = M_Z$.

\par Hence, we have:
\begin{equation}
R_Z\leq \frac{1}{n}\sum_{i=1}^n [I(U_i;Z_i|S_d)-I(S_i;U_i|S_d)]+C_{12}+\epsilon_n. \label{eq5}
\end{equation}
Next, to bound the rate $R_Y$ consider:
\begin{align}
	nR_Y &= H(M_Y)\nonumber\\
    &\stackrel{(a)}{=} H(M_Y|M_Z,S^n)\nonumber\\
	& =H(M_Y|M_Z,S^n)-H(M_Y|M_Z,S^n,Y^n)+H(M_Y|M_Z,S^n,Y^n)\nonumber\\
	& \stackrel{(b)}{\leq}I(M_Y;Y^n|M_Z,S^n)+n\epsilon_n\nonumber\\
    & \stackrel{(c)}{=}I(M_Y,X^n(M_Y,M_Z,S^n);Y^n|M_Z,S^n)+n\epsilon_n\nonumber\\
    & {=} \sum_{i=1}^n I(M_Y,X^n;Y_i|M_Z,S^n,Y^{i-1})+n\epsilon_n\nonumber\\
    & \stackrel{(d)}{=} \sum_{i=1}^n I(M_Y,X^n;Y_i|M_Z,S^n,M_{12},Y^{i-1})+n\epsilon_n\nonumber\\
    & \stackrel{(e)}{=} \sum_{i=1}^n I(M_Y,X^n;Y_i|M_Z,S^n,M_{12},Y^{i-1},Z^{i-1})+n\epsilon_n\nonumber\\
    & {=} \sum_{i=1}^n H(Y_i|M_Z,S^n,M_{12},Y^{i-1},Z^{i-1})-H(Y_i|M_Z,S^n,M_{12},Y^{i-1},Z^{i-1},M_Y,X^n)+n\epsilon_n\nonumber\\
    & \stackrel{(f)}{\leq} \sum_{i=1}^n H(Y_i|M_Z,S_i,S_{d,i})-H(Y_i|M_Z,S^n,M_{12},Y^{i-1},Z^{i-1},M_Y,X^n)+n\epsilon_n\nonumber\\
    & \stackrel{(g)}{\leq} \sum_{i=1}^n H(Y_i|M_Z,S_i,S_{d,i})-H(Y_i|M_Z,S_i,S_{d,i},X_i)+n\epsilon_n\nonumber\\
    & {=} \sum_{i=1}^n I(Y_i;X_i|M_Z,S_{d,i},S_i)+n\epsilon_n\nonumber\\
    & \stackrel{(h)}{=} \sum_{i=1}^n I(Y_i;X_i|U_i,S_{d,i},S_i)+n\epsilon_n\nonumber\\
\end{align}
where\\
$(a)$ follows from the fact that $M_Y$ is independent of $(M_Z,S^n)$,\\
$(b)$ follows from Fano's inequality,\\
$(c)$ follows from the fact that $X^n$ is a deterministic function of $(M_Z,M_Y,S^n)$,\\
$(d)$ follows from the fact that $M_{12}$ is a function of $S^n$,\\
$(e)$ follows from the degradedness property of the channel,\\
$(f)$ follows from the fact that conditioning reduces entropy and the definition of the auxiliary random variable $S_d= (M_{12},S_{i+1}^n,Z^{i-1})$,\\
$(g)$ follows from the properties of the channel,\\
$(h)$ follows from the choice of $U_i = M_Z.$
\par Hence, we have:
\begin{equation}
R_Y\leq \frac{1}{n}\sum_{i=1}^n I(Y_i;X_i|S_i,U_i)+\epsilon_n. \label{eq6}
\end{equation}

\par We complete the proof by using standard time-sharing arguments to obtain the rate bounds terms given in Theorem \ref{Theorem3}.
\end{proof}


\bibliographystyle{plain}
\bibliography{ref}
\end{document}